\documentclass[11pt]{article}

\usepackage[T1]{fontenc}
\usepackage[english]{babel}

\usepackage[a4paper,margin=30mm]{geometry}
\usepackage{graphicx}
\usepackage{amssymb,amsmath,amsthm}
\usepackage{delarray}
\usepackage{enumitem}
\usepackage{tikz}
\usepackage{caption}
\usepackage{hyperref}
\usepackage{authblk}
\usepackage{stmaryrd}
\usepackage{ifthen}
\usepackage{calc}
\usepackage{float}
\usepackage[ruled,vlined,linesnumbered]{algorithm2e}
\usepackage{etoolbox}
\graphicspath{{figures/}}

\newtheorem{theorem}{\textbf{Theorem}}
\newtheorem{lemma}{\textbf{Lemma}}
\newtheorem{proposition}{\textbf{Proposition}}
\newtheorem{corollary}{\textbf{Corollary}}
\newtheorem{remark}{\textbf{Remark}}
\newtheorem{definition}{\textbf{Definition}}
\newtheorem{notation}{\textbf{Notation}}
\newtheorem{example}{\textbf{Example}}

\newcommand{\N}{\mathbb{N}}
\newcommand{\Z}{\mathbb{Z}}
\newcommand{\B}{\mathbb{B}}

\newcounter{finaln}
\newcommand{\entiers}[2][2]{\ifthenelse{\equal{#1}{2}}
{\llbracket #2\rrbracket}
{\ifthenelse{\equal{#1}{0}}{
\ifthenelse{\equal{11}{\the\catcode`#2}}
{\{0,\ldots,#2 -1\}}
{\setcounter{finaln}{#2 -1}
\{0,\ldots,\thefinaln \}}}
{\{1,\ldots,#2\}}}}

\renewcommand{\int}[1]{\entiers{#1}}
\newcommand{\mmjblock}[2]{{#1}_{(#2)}}
\newcommand{\mmjoblock}[2]{{#1}_{\{#2\}}}
\newcommand{\mmjoblocklimit}[2]{{#1}^{\Omega}_{\{#2\}}}
\newcommand{\mmj}[2]{{#1}_{#2}}

\newcommand{\Wl}{(W_\ell)_{\ell \in \entiers{p}}}
\newcommand{\Sk}{\{S_k\}_{k \in \entiers{s}}}

\newcommand{\ispe}{\hat{\imath}}
\newcommand{\jspe}{\hat{\jmath}}
\newcommand{\kspe}{\hat{k}}

\newcommand{\BPn}{\mathsf{BP}_n}
\newcommand{\BP}[1]{\mathsf{BP}_{#1}}
\newcommand{\BSn}{\mathsf{BS}_n}
\newcommand{\BS}[1]{\mathsf{BS}_{#1}}

\newcommand{\lcm}{\text{lcm}}
\renewcommand{\gcd}{\text{gcd}}

\pdfoutput=1

\title{On countings and enumerations\\of block-parallel automata networks}

\author[1,2]{K{\'e}vin Perrot}
\author[1,2]{Sylvain Sen{\'e}}
\author[2]{L{\'e}ah Tapin}
\affil[1]{Universit{\'e} publique, Marseille, France}
\affil[2]{Aix-Marseille Univ, CNRS, LIS, Marseille, France}
\date{}

\begin{document}

\renewcommand{\labelitemi}{$\bullet$}
\renewcommand{\labelitemii}{$\bullet$}
\setlist[itemize,enumerate]{nosep}

\maketitle

\begin{abstract}
	When we focus on finite dynamical systems from both the 
	computability/com\-ple\-xi\-ty and the modelling standpoints, automata 
	networks seem to be a particularly appropriate mathematical model on which 
	theory shall be developed. 
	In this paper, automata networks are finite collections of entities (the 
	automata), each automaton having its own set of possible states, which 
	interact with each other over discrete time, interactions being defined as 
	local functions allowing the automata to change their state according to the 
	states of their neighbourhoods.
	The studies on this model of computation have underlined the very importance 
	of the way (\emph{i.e.}~the schedule) according to which the automata update 
	their states, namely the update modes which can be deterministic, periodic, 
	fair, or not. 
	Indeed, a given network may admit numerous underlying dynamics, these latter 
	depending highly on the update modes under which we let the former evolve. 
	In this paper, we pay attention to a new kind of deterministic, periodic and 
	fair update mode family introduced recently in a modelling framework, called 
	the block-parallel update modes by duality with the well-known and studied 
	block-sequential update modes.
	More precisely, in the general context of automata networks, this work aims 
	at presenting what distinguish block-parallel update modes from 
	block-sequential ones, and at counting and enumerating them: in absolute 
	terms, by keeping only representatives leading to distinct dynamics, and by 
	keeping only representatives giving rise to distinct isomorphic limit 
	dynamics.
	Put together, this paper constitutes a first theoretical analysis of these 
	update modes and their impact on automata networks dynamics.
\end{abstract}

\section{Introduction}

Automata networks were born at the beginning of modern computer science in the 
1940\emph{'s}, notably through the seminal works of McCulloch and 
Pitts~\cite{J-McCulloch1943} on neural networks, and von Neumann on cellular 
automata~\cite{B-vonNeumann1966}, which have become since then widely studied 
models of computation.
The former is classically dived into a finite and heterogeneous structure (a 
graph) whereas the latter is dived into an infinite but regular structure (a 
lattice). 
Whilst there exist deep differences between them, they both belong to the 
family of automata networks, which groups together all the models defined 
locally by means of automata which interact with each 
other over discrete time so that the global computations they operate 
emerge from these local interactions governing them. 
These initial models gave rise to numerous studies around computability 
theory~\cite{CL-Kleene1956,J-Elspas1959,J-Smith1971,B-Conway1982,J-Cook2004} 
and complexity~\cite{J-Kari1994,C-Gamard2021}. 

Beyond these contributions to theoretical computer science, the end of the 
1960s has underlined the prominent role of finite automata networks on 
which we focus in this paper, and Boolean automata networks in particular, in 
the context of biological networks qualitative modelling, thanks to the notable 
works of Kauffman~\cite{J-Kauffman1969} and Thomas~\cite{J-Thomas1973}, who are 
the firsts to have emphasised that the genetic expression profiles can be 
captured by such models, and by the limit behaviours emerging from their 
underlying dynamical systems which can represent for instance phenotypes, 
cellular types, or even biological paces.
Since then, (Boolean) automata networks and extensions of them form the most 
widespread discrete models for gene regulation qualitative 
modelling~\cite{J-Thieffry1995,J-Mendoza1998,J-Shmulevich2002,J-Davidich2008,J-Karlebach2008}. 

From the fundamental standpoint, given an automata network $f$ and two 
different ways of updating the states of its automata over time, \emph{i.e.}~two 
distinct update modes $\mu$ and $\mu'$, the two underlying dynamical systems 
$f_\mu$ and $f_{\mu'}$ can clearly be different. 
In other terms, the update modes play a crucial role on the dynamics of 
automata networks, and acquiring a better understanding of their influence has 
become a hot topic in the domain since Robert's seminal works on discrete 
iterations~\cite{J-Robert1980,B-Robert1986}, leading to numerous further 
studies in the last two 
decades~\cite{J-Goles2008,J-Aracena2009,C-Goles2010,T-Noual2012,J-Aracena2013,J-Noual2018,M-Rios2021,BC-Pauleve2022}. 
This subject is all the more pertinent from both theoretical and applied 
standpoints. 
Indeed, 
\emph{(i)} update modes can be of different natures (deterministic or not, 
periodic or not, ...) and are in infinite and uncountable quantity; and
\emph{(ii)} if we consider automata networks as models of genetic regulation 
networks for instance, we still do not know which ``natural schedules'' govern 
gene expression and regulation even if chromatin dynamics seems to play a key 
role~\cite{J-Hubner2010,J-Fierz2019}.

In~\cite{J-Demongeot2020}, in the framework of biological regulation modelling, 
the authors have introduced a new periodic update mode family, by underlining 
its ``natural'' computational power: this family has features which 
can break the classical property of fixed point set invariance (local update 
repetitions into a period are notably possible) and it allows to capture 
endogenous biological timers/clocks of genetic or physiological nature/origin 
such as those induced with chromatin dynamics. 
Until then, the works addressing the role of periodic update modes focused on 
block-sequential update modes, namely modes in which automata are partitioned 
into a list of subsets such that the automata of a same subset update their 
state all at once while the subsets are iterated sequentially.
Block-parallel update modes are defined dually. 
Rather than being defined as ordered partitions of the set of automata, they 
are defined as sets of lists, or ``partitioned orders'', so that the 
automata of a same list update their state sequentially according to the period 
of the list while the lists are triggered all at once at the initial time step. 

Because they allow local update repetitions which constitute the basis 
for generating non expected dynamical system limit sets and because they seem to 
have a promising role in terms of modelling, we are convinced that they need to 
be addressed in detail, and in different frameworks.
In this paper, we give the first theoretical analysis of combinatorial aspects 
related to block-parallel update modes, in the context of discrete dynamical 
systems underlying automata networks. 
This formal introduction of these update modes \emph{per se} will serve as a 
solid basis for and should pave the way to further developments on the dynamics 
of block-parallel automata networks, which are much more difficult to 
understand than those of classical block-sequential ones.\smallskip

In Section~\ref{s:definitions}, the main definitions and notations which are 
used throughout the paper are presented. 
Section~\ref{s:contributions} develops our main contributions and is divided 
into five parts. 
Section~\ref{s:BSinterBP} deals with the intersection between block-sequential 
and block-parallel update modes; it characterises in particular when a 
block-sequential mode is a block-parallel one and vice-versa.
Section~\ref{s:BP} aims at addressing block-parallel modes in absolute terms; 
we give two neat closed formulas for counting them
(one of which comes from the literature) and provide an algorithm to 
enumerate them.
After having performed numerical simulations of the dynamics of block-parallel 
automata networks, we have observed that, for any automata networks, certain 
block-parallel modes are intrinsically similar in the sense that they always 
lead to exact same dynamics. 
Section~\ref{s:BP0} gives closed formulas for counting intrinsically different 
block-parallel modes in terms of automata network dynamics and an algorithm is 
developed to enumerate them.
In Section~\ref{s:BPstar}, closed formulas and an enumeration algorithm
are provided for intrinsically different block-parallel modes, but up to
isomorphic limit dynamics of automata networks this time.
As numerical experiments exposed in Section~\ref{s:implementations} suggest, it decreases
the number of elements to consider by some order of magnitude,
when one is interested in the asymptotic (limit) behaviour of dynamical systems.
Our code is freely available online (see Section~\ref{s:implementations}).
We conclude this paper by giving perspectives of our work in 
Section~\ref{s:conclusion}.

\section{Definitions}
\label{s:definitions}

Let $\entiers{n} = \entiers[0]{n}$,
let $\B = \{0,1\}$,
let $x_i$ denote the $i$-th component of vector $x \in \B^n$,
let $x_I$ denote the projection of $x$ onto an element of $\B^{|I|}$
for some subset of automata $I \subseteq \entiers{n}$,
let $e_i$ be the $i$-th base vector,
and $\forall x, y \in \B^n$, let $x+y$ denote the bitwise addition modulo two.
Let $\sigma^i$ denote the circular-shift of order $i\in\Z$ on sequences
(shifting the element at position $0$ towards position $i$).
Let $\sim$ denote the graph isomorphism, \emph{i.e.}~for $G = (V,A)$ and 
$G' = (V',A')$ we have $G \sim G'$ if and only if there is a bijection 
$\pi : V \to V'$ such that $(u,v) \in A \iff (\pi(u), \pi(v)) \in A'$.

\paragraph{Automata network} 

An \emph{automata network} (AN) of size $n$ is a discrete dynamical system 
composed of a set of $n$ \emph{automata} $\entiers{n}$, each holding a state 
within a finite alphabet $X_i$ for $i \in \entiers{n}$.
A \emph{configuration} is an element of $X = \prod_{i \in \entiers{n}} X_i$.
An AN is defined by a function $f : X \to X$, decomposed into $n$ 
\emph{local functions} $f_i : X \to X_i$ for $i \in \entiers{n}$, where $f_i$ 
is the $i$-th component of $f$.
To let the system evolve, one must define when the automata update their state
using their local function, which can be done in multiple ways.

\paragraph{Block-sequential update modes} 

A sequence $\Wl$ with $W_\ell \subseteq \entiers{n}$ for all $\ell \in \int{p}$ 
is an \emph{ordered partition} if and only if:
\[
  \bigcup_{\ell \in \int{p}} W_\ell = \entiers{n}
  \text{ and }
  \forall i, j \in \int{p}, i \neq j \implies W_i \cap W_j = \emptyset\text{.}
\]
An update mode $\mu = \Wl$ is called \emph{block-sequential} when
$\mu$ is an ordered partition, and the $W_\ell$ are called \emph{blocks}.
The set of block-sequential update modes of size $n$ is denoted $\BSn$.
The update of $f$ under $\mu \in \BSn$ is given by 
$\mmjblock{f}{\mu} : X \to X$ as follows:
\[
  \mmjblock{f}{\mu}(x) = 
  	\mmjblock{f}{W_{p-1}} \circ \dots \circ \mmjblock{f}{W_1} \circ 
	  \mmjblock{f}{W_0}(x)\text{,}
\]
where for all $\ell \in \int{p}$:
\[
  \forall i \in \entiers{n}, 
  \mmjblock{f}{W_\ell}(x)_i =
	  \begin{cases}
  	  f_i(x) & \text{if } i \in W_\ell\text{,}\\
    	x_i & \text{otherwise.}
	  \end{cases}
\]

\paragraph{Block-parallel update modes}

In a block-sequential update mode, the automata in a block are updated 
simultaneously while the blocks are updating sequentially. 
A block-parallel update mode is based on the dual principle: 
the automata in a block are updated sequentially while the blocks are updated 
simultaneously. 
Instead of being defined as a sequence of unordered blocks, a block-parallel 
update mode will thus be defined as a set of ordered blocks. 
A set $\Sk$ with $S_k = (i^k_0, \dots, i^k_{n_{k}-1})$ a sequence of $n_k>0$ 
elements of $\entiers{n}$ for all $k \in \int{s}$ is a \emph{partitioned order} 
if and only if:
\[
  \bigcup_{k \in \int{s}} S_k = \entiers{n}
  \text{ and }
  \forall i, j \in \int{s}, i \neq j \implies S_i \cap S_j = \emptyset\text{.}
\]
An update mode $\mu = \Sk$ is called \emph{block-parallel} when $\mu$ is a 
partitioned order, and the sequences $S_k$ are called \emph{o-blocks} (for 
\emph{ordered-blocks}).
The set of block-parallel update modes of size $n$ is denoted $\BPn$.
With $p = \lcm(n_1, \dots, n_s)$, the update of $f$ under $\mu \in \BPn$ is 
given by $\mmjoblock{f}{\mu} : X \to X$ as follows:
\[
	\mmjoblock{f}{\mu}(x) = 
		\mmjblock{f}{W_{p-1}} \circ \dots \circ \mmjblock{f}{W_1} \circ
		\mmjblock {f}{W_0}(x)\text{,}
\]
where for all $\ell \in \int{p}$ we define 
$W_\ell = \{i^k_{\ell \mod n_k} \mid k \in \int{s}\}$.

\paragraph{Basic considerations}

As it may sound natural from the definitions above, there is a natural way to 
convert a block-parallel update mode $\Sk$ with 
$S_k = (i^k_0, \dots, i^k_{n_{k-1}})$ into a sequence of blocks of length 
$p = \lcm(n_1, \dots, n_s)$.
We define it as $\varphi$:
\[
  \varphi(\Sk) = 
  	\Wl \text{ with } W_\ell = 
  		\{i^k_{\ell \mod n_k} \mid k \in \int{s}\}\text{.}
\]
In order to differentiate between sequences of blocks and sets of o-blocks,
we denote by $\mmjblock{f}{\mu}$ (resp.~$\mmjoblock{f}{\mu}$) the dynamical 
system induced by $f$ and $\mu$ when $\mu$ is a sequence of blocks (resp.~a set 
of o-blocks), and simply $\mmj{f}{\mu}$ when it is clear from the context.
Moreover, abusing notations, we denote by $\varphi(\BPn)$ the set of 
partitioned orders of $\entiers{n}$ as sequences of blocks.

Block-sequential and block-parallel update modes are \emph{periodic} (the same 
update procedure is repeated at each step), and \emph{fair} (each automaton is 
updated at least once per step).
We distinguish the concepts of \emph{step} and \emph{substep}.
A step is the interval between $x$ and $\mmjblock{f}{\mu}(x)$ (or 
$\mmjoblock{f}{\mu}(x)$), and can be divided into $p = |\mu|$ 
(or $p = |\varphi(\mu)| = \lcm(n_1, \dots, n_s)$) substeps,
corresponding to the elementary intervals in which only one block of automata
is updated.
The most basic update mode is the parallel $\mu_\textsf{par}$ which updates 
simultaneously all automata at each step. 
It is the element $(\entiers{n}) \in \BSn$ and 
$\{(i) \mid i \in \entiers{n}\} \in \BPn$, with 
$\varphi(\{(i) \mid i \in \entiers{n}\}) = (\entiers{n})$.

\begin{remark}\normalfont
	\label{remark:repet}
  Observe that in block-sequential update modes, each automaton is updated 
  exactly once during a step, whereas in block-parallel update modes, some 
  automata can be updated multiple times during a step.
  Update repetitions may have many consequences on the limit dynamics.
  For instance, the network of $n = 3$ automata such that 
  $f_i(x) = x_{i-1 \mod n}$ (\emph{i.e.}~a positive cycle of size $3$) under 
  the update mode $\mu = (\{1,2\}, \{0,2\}, \{0,1\})$ has $4$ fixed points, 
  among which $2$ cannot be obtained with block-sequential update modes
  (in this example, $\mu \notin \BPn$).
\end{remark}

\begin{remark}\normalfont
  \label{remark:BP-bloc-equal-size}
  Let $\mu = \Sk$ be a block-parallel update mode. 
  Each block of $\varphi(\mu)$ is of the same size, namely $s$, and furthermore 
  each block of $\varphi(\mu)$ is unique.
\end{remark}

\paragraph{Fixed points, limit cycles and attractors}

Let $\mmj{f}{\mu}$ be the dynamical system defined by an AN $f$ of size $n$ and
an update mode $\mu$.

Let $p \geq 1$. 
A sequence of configurations $x^0, \ldots, x^{p-1} \in X$ is a \emph{limit 
cycle} of $\mmj{f}{\mu}$ if and only if $\forall i \in \entiers{p},
\mmj{f}{\mu}(x^i) = x^{i+1 \mod p}$.
A limit cycle of length $p = 1$ is a \emph{fixed point}.
The sequence of configurations $x^0, x^1, \ldots, x^{p-1} \in X$ is an 
\emph{attractor} if and only if it is a limit cycle and there exist $x\in X$ 
and $i\in\entiers{p}$ such that $\mmj{f}{\mu}(x) = x^i$ but 
$x \notin \{x^0, \ldots, x^{p-1}\}$.

\begin{example}\normalfont
	\label{ex:updates}
	Let $f : \entiers{3} \times \B \times \B \to \entiers{3} \times \B \times \B$ 
	the automata network defined as:
	\[
		f(x) = \left(
			\begin{array}{ll}
				f_0(x) & = \begin{cases}
					0 & \text{if } ((x_0 = 0) \land (x_1 = x_2)) \lor
							(x_0 = x_1 = x_2 = 1)\\
					1 & \text{if } x_1 + x_2 \mod 2 = 1\\
					2 & \text{otherwise} 
				\end{cases}\\
				f_1(x) & = (x_0 \neq 0) \lor x_1 \lor x_2\\
				f_2(x) & = ((x_0 = 1) \land x_1) \lor (x_0 = 2)\\
			\end{array}
		\right)\text{.}
	\]
	Let $\mu_{\sf bs} = (\{1\}, \{0, 2\})$ and $\mu_{\sf bp} = \{(0), (2, 1)\}$.
	The update mode $\mu_{\sf bs}$ is block-sequential and $\mu_{\sf bp}$ is
	block-parallel, with $\varphi(\mu_{\sf bp}) = (\{0, 2\}, \{0,1\})$.
	Systems $\mmjblock{f}{\mu_{\sf bs}}$ and $\mmjoblock{f}{\mu_{\sf bp}}$ have
	different dynamics, as depicted in Figure~\ref{fig:ex_def}. 
	They both have the same two fixed points and one limit cycle, but the 
	similarities stop there. 
	The limit cycle of $\mmjblock{f}{\mu_{\sf bs}}$ is of size 4, while that of
	$\mmjoblock{f}{\mu_{\sf bp}}$ is of size 2. 
	Moreover, neither of the fixed points of $\mmjoblock{f}{\mu_{\sf bp}}$ is an 
	attractor, while one of $\mmjblock{f}{\mu_{\sf bs}}$, namely $211$, is.
	Both of these update modes' dynamics are unique in $\BP{3} \cup \BS{3}$.

	\begin{figure}
		\centerline{
			\includegraphics[scale=.9]{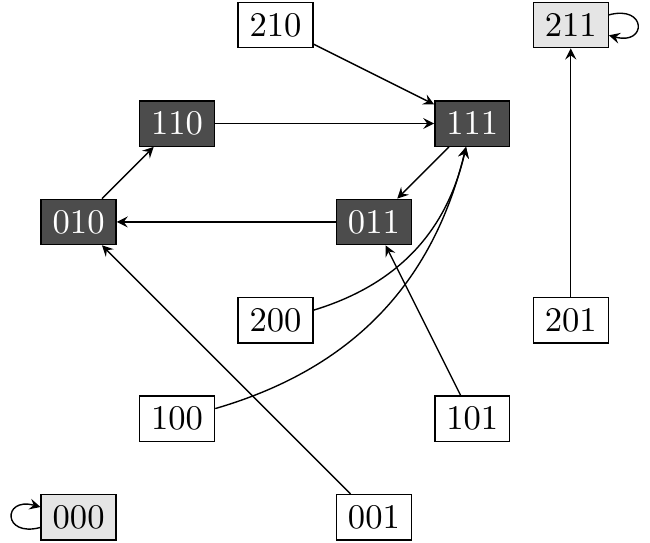}
			\hspace*{1cm}
			\includegraphics[scale=.9]{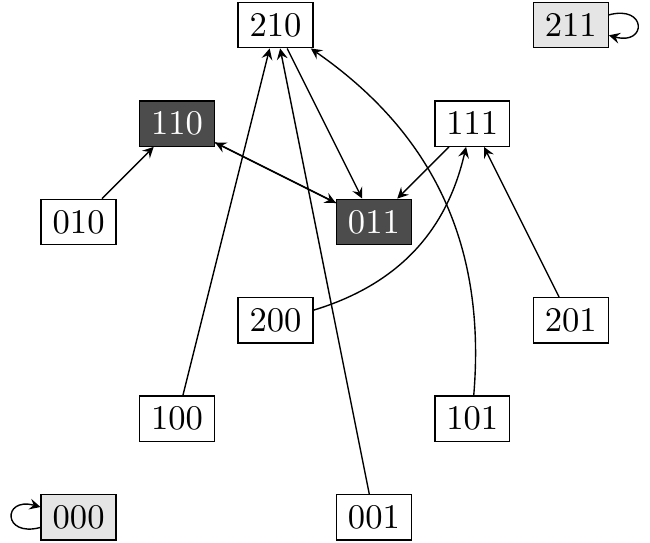}
	  }
  	\caption{The dynamics of $\mmjblock{f}{\mu_{\sf bs}}$ (left) and
	  $\mmjoblock{f}{\mu_{\sf bp}}$ (right) from Example~\ref{ex:updates}.}
  	\label{fig:ex_def}
	\end{figure}
\end{example}

\section{Counting and enumerating block-parallel update modes}
\label{s:contributions}

For the rest of this section,
let $p(n)$ denote the number of integer partitions of $n$
(multisets of integers summing to $n$),
let $d(i)$ be the maximal part size in the $i$-th partition of $n$,
let $m(i,j)$ be the multiplicity of the part of size $j$ in the $i$-th partition 
of $n$.
As an example, let $n = 31$ and assume the $i$-th partition is 
$(2,2,3,3,3,3,5,5,5)$, we have $d(i) = 5$ and $m(i,1) = 0$, $m(i,2) = 2$, 
$m(i,3) = 4$, $m(i,4) = 0$, $m(i,5) = 3$.
A partition will be the support of a partitioned order, where each part is an 
o-block.
In our example, we can have:
\[
  \begin{array}{l}
    \{(0,1),(2,3),(4,5,6),(7,8,9),(10,11,12),(13,14,15),\\
    (16,17,18,19,20),(21,22,23,24,25),(26,27,28,29,30)\}\text{,}
  \end{array}
\]
and we picture it as the following \emph{matrix-representation}:
\[
  \begin{pmatrix}
    0&1\\
    2&3
  \end{pmatrix}
  \begin{pmatrix}
    4&5&6\\
    7&8&9\\
    10&11&12\\
    13&14&15
  \end{pmatrix}
  \begin{pmatrix}
    16&17&18&19&20\\
    21&22&23&24&25\\
    26&27&28&29&30
  \end{pmatrix}
  \text{.}
\]
We call \emph{matrices} the elements of size $j \cdot m(i,j)$ and denote them 
$M_1, \ldots, M_{d(i)}$, where $M_j$ has $m(i,j)$ \emph{rows} and $j$ 
\emph{columns} ($M_j$ is empty when $m(i,j) = 0$).
The partition defines the matrices' dimensions, and each row is an o-block.\medskip

For the comparison, the block-sequential update modes (ordered partitions of 
$\entiers{n}$) are given by the ordered Bell numbers, sequence 
A000670 of OEIS~\cite{oeisA000670,ns11}.
A closed formula for it is:
\[
  |\BS{n}|=
	  \sum_{i = 1}^{p(n)}
  		\frac{n!}{\prod_{j = 1}^{d(i)} 
  			(j!)^{m(i,j)}}
  	\cdot
	  \frac{\left( \sum_{j = 1}^{d(i)} m(i,j) \right)!}
		  {\prod_{j = 1}^{d(i)} m(i,j)!}\text{.}
\]
Intuitively, an ordered partition of $n$ 
gives a support to construct a block-sequential update mode:
place the elements of $\entiers{n}$ up to permutation within the blocks.
This is the left fraction: $n!$ divided by $j!$ for each block of size $j$,
taking into account multiplicities.
The right fraction corrects the count because we sum on $p(n)$ the (unordered) 
partitions of $n$:
each partition of $n$ can give rise to different 
ordered partitions of $n$, 
by ordering all blocks (numerator, where the sum of multiplicities is the number 
of blocks) up to permutation within blocks of the same size which have no effect 
(denominator).
The first ten terms are ($n = 1$ onward):
\[
	1, 3, 13, 75, 541, 4683, 47293, 545835, 7087261, 102247563\text{.}
\]

\subsection{Intersection of block-sequential and block-parallel modes}
\label{s:BSinterBP}

In order to be able to compare block-sequential with block-parallel update
modes, both of them will be written here under their sequence of blocks form
(the classical form for block-sequential update modes and the rewritten form for 
block-parallel modes).

First, we know that $\varphi(\textsf{BP}_n) \cap \textsf{BS}_n$ is not empty, 
since it contains at least 
\[
  \mu_\textsf{par}
  = (\entiers{n})
  = \varphi(\{(0), (1), \ldots, (n-1)\})\text{.}
\]
However, neither $\BSn \subseteq \varphi(\BPn)$ nor $\varphi(\BPn) \subseteq
\BSn$ are true.
Indeed, $\mu_s = (\{0, 1\}, \{2\}) \in \mathsf{BS}_3$ but $\mu_s \notin
\varphi(\mathsf{BP}_3)$ since a block-parallel cannot have blocks of different 
sizes in its sequential form. 
Symmetrically, $\mu_p = \varphi(\{(1, 2), (0)\}) = 
(\{0,1\}, \{0, 2\}) \in \mathsf{BP}_3$ but $\mu_p \notin \mathsf{BS}_3$ since
automaton $0$ is updated twice.
Despite this, we can precisely define the intersection 
$\BSn \cap \varphi(\BPn)$.

\begin{lemma}
	\label{lemma:BPequivBS}
  Let $\mu$ be an update mode written as a sequence of blocks of elements in 
  $\entiers{n}$.  
  Then $\mu \in (\BSn \cap \varphi(\BPn))$ if and only if $\mu$ is 
  an ordered partition and all of $\mu$'s blocks are of the same size.
\end{lemma}

\begin{proof}
  Let $n \in \N$.
  
  $(\Longrightarrow)$ Let $\mu \in (\BSn \cap \varphi(\BPn))$. 
  Since $\mu \in \BSn$, $\mu$ is an ordered partition.
  Furthermore, $\mu \in \varphi(\BPn)$ so all the $\mu$'s blocks are of the same 
  size (Remark~\ref{remark:BP-bloc-equal-size}).

  $(\Longleftarrow)$ Let $\mu = \Wl$ be an ordered partition of $\entiers{n}$ with all its blocks 
  having the same size, denoted by $s$.
  Since $\mu$ is an ordered partition, $\mu \in \BSn$.
  For each $\ell \in \entiers{p}$, we can number arbitrarily the elements of $W_\ell$ 
  from $0$ to $s-1$ as $W_\ell=\{W_\ell^0,\dots,W_\ell^{s-1}\}$.
  Now, let us define the set of sequences $\Sk$ the following way: 
  $\forall k \in \entiers{s}, S_k = 
  \{ W_\ell^k \mid \ell\in\entiers{p} \}$.
  It is a partitioned order such that $\varphi(\Sk) = \mu$, which means
  that $\mu \in \varphi(\BPn)$.
\end{proof}

\begin{corollary}
	\label{cor:BPequivBS}
	If $\mu \in \BPn$ and is composed of $s$ o-blocks of size $p$, then 
	$\varphi(\mu) \in \BSn$ and is composed of $p$ blocks of size $s$.	
\end{corollary}

As a consequence of Lemma~\ref{lemma:BPequivBS} and Corollary~
\ref{cor:BPequivBS}, given $n \in \N$, the set $\mathrm{SEQ}_n$ of sequential 
update modes such that every automaton is updated exactly once by step and only 
one automaton is updated by substep, is a subset of $(\BSn \cap \varphi(\BPn))$.

Moreover, we can state the following proposition which counts the number of 
sequences of blocks which belongs to both $\BSn$ and $\varphi(\BPn)$.

\begin{proposition}
	\label{prop:countBSinterBP}
	Given $n \in \N$, we have:
	\[
		|\BSn\cap \varphi(\BPn)| =
			\sum_{d | n}\frac{n!}{(\frac{n}{d}!)^d}\text{.}
	\]
\end{proposition}

\begin{proof}
	The proof derives directly from the sequence A061095 of 
	OEIS~\cite{oeisA061095}, which counts the \emph{number of ways of dividing $n$ 
	labeled items into labeled boxes with an equal number of items in each box}.
	In our context, the ``items'' are the automata, and the ``labeled boxes'' are 
	the blocks of the ordered partitions.
\end{proof}

\subsection{Partitioned orders}
\label{s:BP}

A block-parallel update mode is given as a partitioned order, \emph{i.e.}~an 
(unordered) set of (ordered) sequences.
This concept is recorded as sequence A000262 of OEIS~\cite{oeisA000262}, 
described as the \emph{number of ``sets of lists''}.
A nice closed formula for it is:
\[
  |\BPn| = 
  	\sum_{i = 1}^{p(n)} 
  		\frac{n!}{\prod_{j=1}^{d(i)} 
  			m(i,j)!}\text{.}
\]
Intuitively, for each partition, fill all the matrices
($n!$ ways to place the elements of $\entiers{n}$)
up to permutation of the 
rows within each matrix
(matrix $M_j$ has $m(i,j)$ rows).
Another closed formula is presented in Proposition~\ref{prop:cardBPn}.
This formula is particularly useful to generate all the block-parallel update 
modes, in the sense that its parts help us construct Algorithm~\ref{algo:BP} 
which enumerates the partitioned orders of $\BPn$.

\begin{proposition}
	\label{prop:cardBPn}
  For any $n \geq 1$ we have:
  \[
  |\BPn| =
	  \sum_{i = 1}^{p(n)} 
	  	\prod_{j = 1}^{d(i)}
	  		\binom{n - \sum_{k = 1}^{j - 1} k \cdot m(i,k)}{j \cdot m(i,j)}
  				\cdot \frac{(j \cdot m(i,j))!}{m(i,j)!}\text{.}
  \]
\end{proposition}

\begin{proof}
	Each partition is a support to generate different partitioned orders (sum on 
	$i$), by considering all the combinations, for each matrix (product on $j$),
	of the ways to choose the $j \cdot m(i,j)$ elements of $\entiers{n}$
	it contains (binomial coefficient, chosen among the remaining elements),
	and all the ways to order them up to permutation of the rows (ratio of 
	factorials).
	Observe that developing the binomial coefficients with
	${\binom{x}{y}} = \frac{x!}{y! \cdot (x-y)!}$ gives
	\[
  	\prod_{j = 1}^{d(i)}
  		\binom{n - \sum_{k}}{j \cdot m(i,j)} \cdot(j \cdot m(i,j))!
		= \prod_{j = 1}^{d(i)}
			\frac{(n - \sum_{k})!}{(n - \sum_{k} - j \cdot m(i,j))!}
		= \frac{n!}{0!}
		= n!\text{,}
	\]
	where $\sum_{k}$ is a shorthand for $\sum_{k = 1}^{j - 1} k \cdot m(i,k)$, 
	which leads to retrieve the OEIS formula.
\end{proof}

The first ten terms are ($n = 1$ onward):
\[
	1, 3, 13, 73, 501, 4051, 37633, 394353, 4596553, 58941091\text{.}
\]

Algorithm~\ref{algo:BP} enumerates the partitioned orders of size $n$.
For each partition $i$ of $n$ (line 2, the partitions of the integer $n$ can be 
enumerated as in~\cite{T-Kelleher2005}, in order to compute $d(i)$ and 
$m(i,j)$), it enumerates all the block-parallel update modes as a list of 
matrices as presented in the introduction of this section.
The matrices are filled one-by-one (argument $j$ of the auxiliary function)
using recursion (when a matrix is ready it makes a call on $j+1$ at line 18 
or 20).
When $d(i) < j$, the block-parallel update mode is complete and enumerated 
(line 6).
When $m(i,j) = 0$, the matrix $M_j$ is empty and the recursive call is immediate 
(line 20).
Otherwise, we choose a set $A$ of $j \cdot m(i,j)$ elements of $\entiers{n}$ to 
put in matrix $M_j$, \emph{i.e.}~a combination among the 
$\binom{n}{j \cdot m(i,j)}$ possible ones (line 9).
In order to enumerate all such matrices up to permutation of the rows, we choose 
the elements of the first column (line 10), and place them in a predefined order 
(lines 11-13).
The rest of the matrix is filled in all the possible ways (lines 14-17) and a 
recursive call is performed for each possibility (line 18), in order to repeat 
the procedure on the remaining matrices with the remaining elements of 
$\entiers{n}$.

\begin{algorithm}[t!]
  {\small \caption{Enumeration of $\BPn$}
  \label{algo:BP}
  \DontPrintSemicolon
  \SetKwProg{Fn}{Function}{:}{}
  \SetKwFunction{FEnumBP}{EnumBP}
  \SetKwFunction{FEnumBPaux}{EnumBPaux}
  \SetKwComment{Comment}{\# }{}
	\bigskip  
  
  \Fn{\FEnumBP{$n$}}{
    \ForEach{$i$ from $1$ to $p(n)$}{
      \FEnumBPaux{$n$,$i$,$1$}
    }
  }\medskip

  \Fn{\FEnumBPaux{$n$, $i$, $j$, $M_1$, \dots, $M_{j-1}$}}{
    \If{$d(i)<j$}{
      \texttt{enumerate(}$M$\texttt{)}\\
      \Return
    }
    \If{$m(i, j) > 0$}{ 
      \ForEach{combination $A$ of size $j\cdot m(i, j)$ among 
      	$\entiers{n} \setminus \bigcup_{k=1}^{j-1}M_k$}{
        \ForEach{combination $C$ of size $m(i, j)$ among $A$}{
          \texttt{sort(}$C$\texttt{)}\\
          \For{$k\leftarrow 1$ \KwTo $m(i,j)$}{
            $M_j[k][1]\leftarrow C[k]$
          }
          \ForEach{permutation $B$ of $A\setminus C$}{
            \For{$k\leftarrow 1$ \KwTo $m(i,j)$}{
              \For{$\ell\leftarrow 2$ \KwTo $j$}{
                $M_j[k][\ell]\leftarrow B[(k-1)\cdot(j-1)+\ell-1]$
              }
            }
            \FEnumBPaux{$n$,$i$,$j+1$, $M_1$,\dots,$M_{j-1}$,$M_j$}
          }
        }
      }
    }
    \Else{
      \FEnumBPaux{$n$,$i$,$j+1$, $M_1$,\dots,$M_{j-1}$,$\emptyset$}
    }
  }}
\end{algorithm}

\begin{theorem}
  Algorithm~\ref{algo:BP} is correct: it enumerates $\BPn$.
\end{theorem}

\begin{proof}
	Firstly, it enumerates the correct number of update modes for $n$.
	Indeed, for each partition of $n$ and for each $j$ from $1$ to $d(i)$,
	we choose the $j \cdot m(i,j)$ elements to fill matrix $M_j$.
	Then, we choose the $m(i,j)$ elements of the first column among them, and take 
	all permutations for the remaining elements of this matrix.
	This gives the following equation:
	\begin{align*}
		|\texttt{EnumBP(}n\texttt{)}\hspace*{-1.7pt}|
			& \;=\; \sum_{i = 1}^{p(n)} \prod_{j = 1}^{d(i)}
					{\binom{n - \sum_{k = 1}^{j-1} k \cdot m(i,j)}{j \cdot m(i,j)}}
			  	\binom{m(i,j) \cdot j}{m(i, j)} (m(i,j) (j-1))!\\
			& \;=\; \sum_{i = 1}^{p(n)} \prod_{j = 1}^{d(i)}
				  {\binom{n - \sum_{k = 1}^{j-1} k \cdot m(i,j)}{j \cdot m(i,j)}} 
				  \frac{m(i,j) \cdot j!}{(m(i,j))!}\\
			& \;=\; |\BPn|\text{.}
	\end{align*}

	Secondly, since the first column of each matrix is always in ascending order, 
	the algorithm cannot enumerate two matrix representations that are identical 
	up to permutation of the rows. 
	Thus, it cannot return the same partitioned order twice, which means that it 
	is indeed enumerating every partitioned order.
\end{proof}

\subsection{Partitioned orders up to dynamical equality}
\label{s:BP0}

As for block-sequential update modes, given an AN $f$ and two block-parallel 
update modes $\mu$ and $\mu'$, the dynamics of $f$ under $\mu$ can be the same 
as that of $f$ under $\mu'$. 
To go further, in the framework of block-parallel update modes,
there exist pairs of update modes $\mu,\mu'$ such that for any AN $f$,
the dynamics $\mmjoblock{f}{\mu}$ is the exact same as $\mmjoblock{f}{\mu'}$.
As a consequence, in order to perform exhaustive searches among
the possible dynamics, it is not necessary to generate all of them.
We formalize this with the following equivalence relation.

\begin{definition}
  \label{def:dyn_eq}
  For $\mu,\mu'\in\BPn$, we denote $\mu\equiv_0\mu'$
  when $\varphi(\mu)=\varphi(\mu')$.
\end{definition}

The following Lemma shows that this equivalence relation is necessary and 
sufficient in the general case of ANs of size $n$.

\begin{lemma}
  For any $\mu, \mu' \in \BPn$, we have $\mu \equiv_0 \mu' \iff 
  \forall f: X \to X, \mmjoblock{f}{\mu} = \mmjoblock{f}{\mu'}$.
\end{lemma}

\begin{proof}
	Let $\mu$ and $\mu'$ be two block-parallel update modes of $\BPn$.\medskip
	
  \noindent $(\Longrightarrow)$ Let us consider that $\mu\equiv_0\mu'$, 
  and let $f : X \to X$ be an AN. 
  Then, we have $\mmjoblock{f}{\mu} = \mmjblock{f}{\varphi(\mu)} =
  \mmjblock{f}{\varphi(\mu')} = \mmjoblock{f}{\mu'}$.\medskip

  \noindent $(\Longleftarrow)$ Let us consider that 
  $\forall f : X \to X, \mmjoblock{f}{\mu} = \mmjoblock{f}{\mu'}$. 
  Let us assume for the sake of contradiction that $\varphi(\mu) \neq
  \varphi(\mu')$. 
  For ease of reading, we will denote as $t_{\mu, i}$ the substep at which
  automaton $i$ is updated for the first time with update mode $\mu$. 
  Then, there is a pair of automata $(i,j)$ such that 
  $t_{\mu, i} \leq t_{\mu, j}$, but $t_{\mu', i} > t_{\mu',j}$. 
  Let $f : \B^n \to \B^n$ be a Boolean AN such that $f(x)_i = x_i \lor
	x_j$ and $f(x)_j = x_i$, and $x \in \B^n$ such that $x_i = 0$ and $x_j = 1$. 
	We will compare $\mmjoblock{f}{\mu}(x)_i$ and $\mmjoblock{f}{\mu'}(x)_i$, in 
	order to prove a contradiction.
	Let us apply $\mmjoblock{f}{\mu}$ to $x$. 
	Before step $t_{\mu, i}$ the value of automaton $i$ is still $0$ and, most 
	importantly, since $t_{\mu, i} \leq t_{\mu,j}$, the value of $j$ is still $1$. 
	This means that right after step $t_{\mu,i}$, the value of automaton $i$ is 1, 
	and will not change afterwards. 
	Thus, we have $\mmjoblock{f}{\mu}(x)_i = 1$.
	Let us now apply $\mmjoblock{f}{\mu'}$ to $x$. 
	This time, $t_{\mu', i} > t_{\mu', j}$, which means that automaton $j$ is 
	updated first and takes the value of automaton $i$ at the time, which is $0$ 
	since it has not been updated yet. 
	Afterwards, neither automata will change value since $ 0 \lor 0$ is still $0$.
	This means that $\mmjoblock{f}{\mu'}(x)_i = 0$.
	Thus, we have $\mmjoblock{f}{\mu} \neq \mmjoblock{f}{\mu'}$, which contradicts
	our earlier hypothesis.
\end{proof}

Let $\BPn^0=\BPn/\equiv_0$ denote the corresponding quotient set, 
\emph{i.e.}~the set of block-parallel update modes to generate for computer 
analysis of all the possible dynamics in the general case of ANs of size $n$.

\begin{theorem}
  \label{theorem:BPn0}
  For any $n\geq 1$, we have:
	\begin{align}
		|\BPn^0| 
			& \;=\; 
				\sum_{i = 1}^{p(n)} 
					\frac{n!}{\prod_{j = 1}^{d(i)} \left( m(i,j)! \right)^j}
					\label{BPn0_eq1}\\
			& \;=\; 
				\sum_{i = 1}^{p(n)} 
					\prod_{j = 1}^{d(i)} 
						\prod_{\ell = 1}^{j}
							\binom{n - \sum_{k = 1}^{j - 1} 
								k \cdot m(i,k) - (\ell - 1) \cdot m(i,j)}{m(i,j)}
								\label{BPn0_eq2}\\
			& \;=\; 
				\sum_{i = 1}^{p(n)} 
					\prod_{j = 1}^{d(i)} 
						\left( 
							\binom{n - \sum_{k = 1}^{j - 1} k \cdot m(i,k)}{j \cdot m(i,j)}
							\cdot 
							\prod_{\ell = 1}^{j} 
								\binom{(j - \ell + 1) \cdot m(i,j)}{m(i,j)}
						\right)\label{BPn0_eq3}\text{.}
	\end{align}
\end{theorem}

\begin{proof}
	$|\BPn^0|$ can be viewed as three distinct formulas. 
	Let us begin this proof by showing that these formulas are equal, then we 
	will show that they effectively count $|\BPn^0|$.
	
  Formula~\ref{BPn0_eq1} is a sum for each partition of $n$ (sum 
  on $i$), of all the ways to fill all the matrices ($n!$) up to permutation
  within each column ($m(i,j)!$ for each of the $j$ columns of $M_j$).

  Formula~\ref{BPn0_eq2} is a sum for each partition of $n$ (sum on $i$),
  of the product for each column of the matrices (products on $j$ and $\ell$),
  of the choice of elements (among the remaining ones) to fill the column
  (regardless of their order within the column).

  Formula~\ref{BPn0_eq3} is a sum for each partition of $n$ (sum on $i$),
  of the product for each matrix (product on $j$), of the choice of elements
  (among the remaining ones) to fill this matrix, multiplied by the number of 
  ways to fill the columns of the matrix (product on $\ell$) with these elements
  (regardless of their order within each column).

  The equality between Formulas~\ref{BPn0_eq1} and~\ref{BPn0_eq2} is obtained 
  by developing the binomial coefficients as follows: 
  ${\binom{x}{y}} = \frac{x!}{y!\cdot(x-y)!}$, 
  and by observing that the products of $\frac{x!}{(x-y)!}$ telescope.
  Indeed, denoting 
  $a(j,\ell) = (n - \sum_{k = 1}^{j-1} k \cdot m(i,k) - \ell \cdot m(i,j))!$,
  we have
  \[
    \prod_{j = 1}^{d(i)} 
    	\prod_{\ell = 1}^{j}
				\frac
					{(n - \sum_{k = 1}^{j - 1} k \cdot m(i,k) - (\ell - 1) \cdot m(i,j))!}
					{(n - \sum_{k = 1}^{j - 1} k \cdot m(i,k) - \ell \cdot m(i,j))!}
    \;=\;
    \prod_{j = 1}^{d(i)} 
    	\prod_{\ell = 1}^{j}
    		\frac
					{a(j,\ell-1)}
					{a(j,\ell)}
    \;=\;
    \frac{n!}{0!}
    \;=\;
    n!
  \]
  because $a(1,0) = n!$, then $a(1,j) = a(2,0)$, $a(2,j) = a(3,0)$, ..., until 
  $a(d(i),j) = 0!$.\medskip

  The equality between Formulas~\ref{BPn0_eq2} and~\ref{BPn0_eq3} is obtained by 
  repeated uses of the identity
  ${\binom{x}{z}}{\binom{x-z}{y}} = {\binom{x}{z+y}}{\binom{z+y}{y}}$, 
  which gives by induction on $j$:
  \begin{align}
  	\label{eq:binom23}
    \prod_{\ell = 1}^{j} 
    	\binom{x - (\ell - 1) \cdot y}{y}
    \;=\;
    \binom{x}{j \cdot y} \cdot \prod_{\ell = 1}^{j} 
    	\binom{(j - \ell + 1) \cdot y}{y}\text{.}
  \end{align}
  Indeed, $j=1$ is trivial and, using the induction hypothesis on $j$ then the 
  identity we get:
  \begin{align*}
    \prod_{\ell = 1}^{j+1} 
    	\binom{x - (\ell - 1) \cdot y}{y}
    & \;=\; 
    	\binom{x - j \cdot y}{y} \cdot \prod_{\ell = 1}^{j} 
	    	\binom{x - (\ell - 1) \cdot y}{y}\\
    & \;=\; 
    	\binom{x - j \cdot y}{y} \cdot \binom{x}{j \cdot y} 
    	\cdot \prod_{\ell = 1}^{j} 
    		\binom{(j - \ell + 1) \cdot y}{y}\\
    & \;=\;
    	\binom{x}{(j + 1) \cdot y} \cdot \binom{(j + 1) \cdot y}{y} 
    	\cdot \prod_{\ell = 1}^{j} 
    		\binom{(j - \ell + 1) \cdot y}{y}\\
    & \;=\;
	    \binom{x}{(j + 1) \cdot y} \cdot 
	    \prod_{\ell = 0}^{j} 
	    	\binom{(j - \ell + 1) \cdot y}{y}\\
    & \;=\;
    	\binom{x}{(j + 1) \cdot y} \cdot \prod_{\ell = 1}^{j + 1} 
    		\binom{(j + 1 - \ell + 1) \cdot y}{y}\text{.}
  \end{align*}
  As a result, Formula~\ref{BPn0_eq3} is obtained from Formula~\ref{BPn0_eq2} 
  by applying Equation~\ref{eq:binom23} for each $j$ with 
  $x = n - \sum_{k = 1}^{j-1} k \cdot m(i,k)$ and $y = m(i,j)$.\medskip

  To prove that they count the number of elements in $\BPn^0$, we now argue that 
  for any pair $\mu, \mu' \in \BPn$, we have $\mu \equiv_0 \mu'$ if and only if
  their matrix-representations are the same up to a permutation of the elements 
  within columns (the number of equivalence classes is then counted by Formula~\ref{BPn0_eq1}).
  In the definition of $\varphi$, each block is a set constructed by taking one 
  element from each o-block. 
  Given that $n_k$ in the definition of $\varphi$ corresponds to $j$ in the 
  statement of the theorem, one matrix corresponds to all the o-blocks that have 
  the same size $n_k$.
  Hence, the $\ell \mod n_k$ operations in the definition of $\varphi$
  amounts to considering the elements of these o-blocks which are in the 
  same column in their matrix representation.
  Since blocks are unordered, the result follows.
\end{proof}

The first ten terms of the sequence $(|\BPn^0|)_{n\geq 1}$ are:
\[
	1, 3, 13, 67, 471, 3591, 33573, 329043, 3919387, 47827093\text{.}
\]
They match the sequence A182666 of OEIS~\cite{oeisA182666},
and the next lemma proves that they are indeed the same sequence
(defined by its exponential generating function on OEIS).
The \emph{exponential generating function of a sequence}
$\left(a_n\right)_{n\in\N}$ is $f(x) = \sum_{n\geq 0}a_n\frac{x^n}{n!}$.

\begin{lemma}
  The exponential generating function of $(|\BPn^0|)_{n \in \N}$ is
  $\prod_{j \geq 1} \sum_{k \geq 0} \left( \frac{x^k}{k!} \right)^j$.
\end{lemma}

\begin{proof}
	We will start from the exponential generating function by finding the
	coefficient of $x^n$ and proving that it is equal to $\frac{|\BPn^0|}{n!}$, 
	and thus that the associated sequence is $(|\BPn^0|)_{n\in\N}$.

	\begin{multline*}
		\prod_{j \geq 1} 
			\sum_{k \geq 0} 
				\left( 
					\frac{x^k}{k!} 
				\right)^j
		\;=\;
			\left( 
				\sum_{k \geq 0} 
					\frac{x^k}{k!} 
			\right) 
			\times 
			\left( 
				\sum_{k \geq 0} 
					\frac{x^{2k}}{(k!)^2} 
			\right) 
			\times 
			\left( 
				\sum_{k \geq 0} 
					\frac{x^{3k}}{(k!)^3} 
			\right) 
			\times  
			\cdots \\
		\;=\; 
			\underbrace{
				\left( 1 + x + \frac{x^2}{2!} + \cdots \right)
			}_{j = 1} 
			\times
			\underbrace{
				\left( 1 + x^2 + \frac{x^4}{(2!)^2} + \cdots \right)
			}_{j = 2}
			\times
			\underbrace{
				\left( 1 + x^3 + \frac{x^6}{(2!)^3} + \cdots \right)
			}_{j = 3}
			\times
			\cdots
			\text{.}
	\end{multline*}
	Each term of the distributed sum is obtained by associating a $k \in \N$ to 
	each $j \in \N_+$, and by doing the product of the $\frac{1}{(k!)^j} \cdot 
	x^{jk}$.
	Thus, if $\N^{\N_+}$ is the set of maps from $\N_+$ to  $\N$, we have:
	\[
		\prod_{j \geq 1} 
			\sum_{k \geq 0} \left( 
				\frac{x^k}{k!} 
			\right)^j 
		\;=\;
			\sum_{m \in \N^{\N_+}} 
				\left( 
					\prod_{j \geq 1} 
						\frac{1}{(m(j)!)^j} 
				\right)
				\cdot 
				x^{\sum_{j \geq 1} j \cdot m(j)}\text{.}
	\]
	From here, to get the coefficient of $x^n$, we need to do the sum only on the
	maps $m$ such that $\sum_{j\geq1}j\cdot m(j) = n$, which just so happen to be
	the partitions of $n$, with $m(j)$ being the multiplicity of $j$ in the
	partition.
	Thus, the coefficient of $x^n$ is
	\[
  	\sum_{i = 1}^{p(n)} 
  		\prod_{j \geq 1} 
  			\frac{1}{(m(i,j)!)^j} 
	  \;=\; 
	  	\sum_{i = 1}^{p(n)} 
	  		\frac{1}{\prod_{j \geq 1}^{d(i)} (m(i,j)!)^j} 
	  \;=\; 
	  	\frac{|\BPn^0|}{n!}\text{.}
	\]
\end{proof}

Algorithm~\ref{algo:BP0} enumerates the elements of $\BPn^0$. It starts out
like Algorithm~\ref{algo:BP}, but differs from line 10 and onwards, after the
contents of $M_j$ are chosen. In the previous algorithm, we needed to enumerate
every matrix up to permutation of the rows. This time, we need to enumerate
every matrix up to permutation within the columns. This means that we only
need to choose the content of each column, not the order. This is performed
using two other functions to enumerate the possible contents for each column
recursively, just like we do for that of the matrices.
The first function (lines 14-20) works in an analogous way to function
\texttt{EnumBPeq0}, minus the enumeration of partitions ; it is mostly there to
reshape the columns given by the auxiliary function as a matrix.
The auxiliary function either returns the columns if every element of the
matrix has been placed into one of them (lines 22-24), or enumerates the 
possible ways to fill the current column, and calls itself recursively to fill 
the next one in all possible ways as well (lines 25-27).

\begin{algorithm}[t!]
  {\small \caption{Enumeration of $\BPn^0$}
  \label{algo:BP0}
  \DontPrintSemicolon
  \SetKwProg{Fn}{Function}{:}{}
  \SetKwFunction{FEnumBP}{EnumBPeq0}
  \SetKwFunction{FEnumBPaux}{EnumBPeq0aux}
  \SetKwFunction{FEnumBlock}{EnumBlockeq0}
  \SetKwFunction{FEnumBlockAux}{EnumBlockeq0Aux}
  \SetKwComment{Comment}{\# }{}
	\bigskip  
  
  \Fn{\FEnumBP{$n$}}{
    \ForEach{$i \in \emph{partitions}(n)$}{
      \FEnumBPaux{$n$,$i$,$1$}
    }
  }\medskip

  \Fn{\FEnumBPaux{$n$, $i$, $j$, $M_1$, \dots, $M_{j-1}$}}{
    \If{$d(i)<j$}{
      \texttt{enumerate(}$M$\texttt{)}\;
      \Return
    }
    \If{$m(i, j) > 0$}{ 
      \ForEach{combination $A$ of size $j\cdot m(i, j)$ among
				$\entiers{n}\setminus \bigcup_{k=1}^{j-1}M_k$}{
        \ForEach{$M_j$ enumerated by \FEnumBlock{$A$, $j$, $m(i,j)$}}{
          \FEnumBPaux{$n$, $i$, $j+1$, $M_1$, \dots, $M_{j-1}$, $M_j$}
        }
      }
    }
    \Else{
      \FEnumBPaux{$n$, $i$, $j+1$, $M_1$, \dots, $M_{j-1}$, $\emptyset$}
    }
  }\medskip

  \Fn{\FEnumBlock{$A$, $j$, $m$}}{
    \ForEach{$C$ enumerated by \FEnumBlockAux{$A$, $m$}}{
      \For{$k\leftarrow 1$ \KwTo $m$}{
        \For{$\ell\leftarrow 1$ \KwTo $j$}{
          $M_j[k][\ell] = C_\ell[k]$\;

        }
      }
      \texttt{enumerate(}$M_j$\texttt{)}\;
      \Return
    }
  }\medskip

  \Fn{\FEnumBlockAux{$A$, $m$, $C_1$, \dots, $C_\ell$}}{
    \If{$A = \emptyset$}{
      \texttt{enumerate(}$C$\texttt{)}\;
      \Return
    }
    \Else{
      \ForEach{combination $B$ of size $m$ among $A$}{
        \FEnumBlockAux{$A\setminus B$, $m$, $C_1$, \dots, $C_\ell$, $B$}\;
      }
    }
  }}
\end{algorithm}

\begin{theorem}
  Algorithm~\ref{algo:BP0} is correct: it enumerates $\BPn^0$.
\end{theorem}

\begin{proof}
	Firstly, it enumerates the correct number of update modes for $n$.
	Indeed, for each partition of $n$ and for each $j$ from $1$ to $d(i)$,
	we first choose the $j\cdot m(i,j)$ elements to fill matrix $M_j$. 
	We then choose, for each $\ell$ from $1$ to $j$, the $m(i, j)$ elements of the
	$\ell$-th column of $M_j$. Thus, the algorithm enumerates the following number
	of update modes:
	\[
		\sum_{i = 1}^{p(n)}
			\prod_{j = 1}^{d(i)}
				\left( 
					\binom{n - \sum_{k = 1}^{j-1} k \cdot m(i,k)}
					{j \cdot m(i,j)} 
					\cdot
					\prod_{\ell = 1}^{j} 
						\binom{(j - \ell + 1) \cdot m(i,j)}{m(i,j)}
				\right)\text{,}
	\] 
	which is exactly the size of $\BPn^0$ (Theorem~\ref{theorem:BPn0}).

	Secondly, the algorithm only chooses the content of each matrix column, not 
	the order. 
	This means that two update modes enumerated by this algorithm cannot have the 
	exact same content of each column of each of their matrices, and thus cannot 
	be equivalent by Definition~\ref{def:dyn_eq}.
	This means that Algorithm~\ref{algo:BP0} enumerates every block-parallel
	update mode up to dynamical equality.
\end{proof}

\subsection{Partitioned orders up to dynamical isomorphism on the limit set}
\label{s:BPstar}

The following equivalence relation defined over block-parallel update modes
turns out to capture exactly the notion of having isomorphic limit dynamics.
It is analogous to $\equiv_0$, except that a circular shift of order $i$
may be applied on the sequences of blocks.

\begin{definition}
\label{def:dyn_iso}
  For $\mu, \mu' \in \BPn$, we denote $\mu \equiv_\star \mu'$ when 
  $\varphi(\mu) = \sigma^i(\varphi(\mu'))$ for some $i \in 
  \entiers{|\varphi(\mu')|}$ called the \emph{shift}.
\end{definition}

\begin{remark}\normalfont
  Note that $\mu \equiv_0 \mu'$ corresponds to the particular case $i = 0$ of 
  $\equiv_\star$.
  Thus, $\mu \equiv_0 \mu' \implies \mu \equiv_\star \mu'$.
\end{remark}

\begin{notation}
  Given $\mmjoblock{f}{\mu} : X \to X$, let
	$\Omega_{\mmjoblock{f}{\mu}} = \bigcap_{t \in \N} \mmjoblock{f}{\mu}(X)$
  denote its \emph{limit set} (abusing the notation of $\mmjoblock{f}{\mu}$ to 
  sets of configurations), 
  and $\mmjoblocklimit{f}{\mu} : \Omega_{\mmjoblock{f}{\mu}} \to
  \Omega_{\mmjoblock{f}{\mu}}$ its restriction to its limit set.
  Observe that, since the dynamics is deterministic, $\mmjoblocklimit{f}{\mu}$ 
  is bijective.
\end{notation}

The following Lemma shows that, if one is generally interested in the limit 
behaviour of ANs under block-parallel updates, then studying a representative 
from each equivalence class of the relation $\equiv_\star$ is necessary and 
sufficient to get the full spectrum of possible limit dynamics.

\begin{lemma}
  For any $\mu, \mu' \in \BPn$, we have $\mu \equiv_\star \mu' \iff
  \forall f : X \to X, \mmjoblocklimit{f}{\mu} \sim \mmjoblocklimit{f}{\mu'}$.
\end{lemma}

\begin{proof}
	Let $\mu$ and $\mu'$ be two block-parallel update modes of $\BPn$.\medskip
	
  \noindent $(\Longrightarrow)$ Let $\mu,\mu'$ be such that $\mu \equiv_\star 
  \mu'$ of shift $\ispe \in \entiers{p}$, with $\varphi(\mu) = \Wl$, 
  $\varphi(\mu') = (W'_\ell)_{\ell \in \entiers{p}}$ and $p = |\varphi(\mu)| = 
  |\varphi(\mu')|$.
  It means that $\forall i \in \entiers{p}$, we have $W'_i = 
  W_{i + \ispe \mod p}$, and for any AN $f$, we deduce that 
  $\pi = \mmjblock{f}{W_0, \ldots, W_{\ispe-1}}$ is the desired isomorphism from 
  $\Omega_{\mmjoblock{f}{\mu}}$ to $\Omega_{\mmjoblock{f}{\mu'}}$.
  Indeed, we have $\mmjoblock{f}{\mu}(x) = y$ if and only if 
  $\mmjoblock{f}{\mu'}(\pi(x)) = \pi(y)$ because
  \[
    \mmjoblock{f}{\mu'} \circ \pi
    =
    \mmjblock{f}{W_0, \ldots, W_{\ispe-1}, W'_0, \ldots, W'_p}
    =
    \mmjblock{f}{W'_{p-\ispe}, \ldots, W'_p} \circ \mmjoblock{f}{\mu}
    =
    \pi \circ \mmjoblock{f}{\mu}\text{.}
  \]
  Note that $\pi^{-1} = \mmjoblock{f}{\mu'}^{(q-1)} \circ 
  \mmjblock{f}{W'_{\ispe} \ldots W'_{p-1}}$ with $q$ the least common multiple 
  of the limit cycle lengths, and $\pi^{-1} \circ \pi$ (\emph{resp.} 
  $\pi \circ\pi^{-1}$)  is the identity on $\Omega_{\mmjoblock{f}{\mu}}$
  (\emph{resp.} $\Omega_{\mmjoblock{f}{\mu'}}$).\medskip

  \noindent $(\Longleftarrow)$ We prove the contrapositive, from $\mu \not 
  \equiv_\star \mu'$, by case disjunction.\smallskip
  
  \begin{enumerate}
  \item[(1)] If in $\varphi(\mu)$ and $\varphi(\mu')$, there is an automaton 
	  $\ispe$ which is not updated the same number of times $\alpha$ and $\alpha'$ 
  	in $\mu$ and $\mu'$ respectively, then we assume without loss of generality 
	  that $\alpha > \alpha'$ and consider the AN $f$ such that:
  	\begin{itemize}
    \item $X_{\ispe} = \entiers{\alpha}$ and $X_i = \{0\}$ for all 
    	$i \neq \ispe$; and
    \item $f_{\ispe}(x) = (x_{\ispe} + 1) \mod \alpha$ and $f_i(x) = x_i$ for 
    	all $i \neq \ispe$.
	  \end{itemize}
  	It follows that $\mmjoblocklimit{f}{\mu}$ has only fixed points since 
  	$+1 \mod \alpha$ is applied $\alpha$ times, whereas 
  	$\mmjoblocklimit{f}{\mu'}$ has no fixed point because $\alpha' < \alpha$.
	  We conclude that $\mmjoblocklimit{f}{\mu} \not \sim 
	  \mmjoblocklimit{f}{\mu'}$.\smallskip
	  
	\item[(2)] If in $\varphi(\mu)$ and $\varphi(\mu')$, all the automata are 
		updated the same number of times, then the transformation from $\mu$ to 
		$\mu'$ is a permutation on $\entiers{n}$ which preserves the matrices of 
		their matrix representations (meaning that any $i \in \entiers{n}$ is in an 
		o-block of the same size in $\mu$ and $\mu'$, which also implies that $\mu$ 
		and $\mu'$ are constructed from the same partition of $n$). 
		Then we consider subcases.
		
		\begin{enumerate}
		\item[(2.1)] If one matrix of $\mu'$ is not obtained by a permutation of the 
			columns from $\mu$, then there is a pair of automata $\ispe, \jspe$
			that appears in the $k$-th block of $\varphi(\mu)$ for some $k$,
			and does not appear in any block of $\varphi(\mu')$.
			Indeed, one can take $\ispe, \jspe$ to be in the same column in $\mu$ but 
			in different columns in $\mu'$.
			Let $S$ be the o-block of $\ispe$ and $S'$ be the o-block of $\jspe$.
			Let $p$ denote the least common multiple of o-blocks sizes in both $\mu$ 
			and $\mu'$.
			In this case we consider the AN $f$ such that:
			\begin{itemize}
			\item $X_{\ispe} = \B \times \entiers{\frac{p}{|S|}}$,
	      $X_{\jspe} = \B \times \entiers{\frac{p}{|S'|}}$,
				and $X_i = \{0\}$ for all $i \notin \{\ispe,\jspe\}$.
				Given $x \in X$, we denote $x_{\ispe} = (x_{\ispe}^b, x_{\ispe}^\ell)$ 
				the state of $\ispe$ (and analogously for $\jspe$); and
			\item $f_{\ispe}(x) = \begin{cases}
        	(x_{\jspe}^b, x_{\ispe}^\ell + 1 \mod \frac{p}{|S|}) 
        		& \text{if } x_{\ispe}^\ell = 0\\
	        (x_{\ispe}^b, x_{\ispe}^\ell + 1 \mod \frac{p}{|S|}) 
  	      	& \text{otherwise}
    	  \end{cases}$,\\
	      $f_{\jspe}(x) = \begin{cases}
	        (x_{\ispe}^b, x_{\jspe}^\ell + 1 \mod \frac{p}{|S'|}) 
	        	& \text{if } x_{\jspe}^\ell = 0\\
        	(x_{\jspe}^b, x_{\jspe}^\ell + 1 \mod \frac{p}{|S'|}) 
        		& \text{otherwise}
	      \end{cases}$, and\\
	      $f_i(x) = x_i$ for all $i \notin \{\ispe,\jspe\}$.
 			\end{itemize}
			Note that $\ispe$ (resp.~$\jspe$) is updated $\frac{p}{|S|}$ 
			(resp.~$\frac{p}{|S'|}$) times during a step in both $\mu$ and $\mu'$.
			Therefore for any $x \in X$, its two images under $\mu$ and $\mu'$ verify
			$\mmjoblock{f}{\mu}(x)^\ell_{\ispe} = \mmjoblock{f}{\mu'}(x)^\ell_{\ispe} 
				= x^\ell_{\ispe}$ (and analogously for $\jspe$).
			Thus for the evolution of the states of $\ispe$ and $\jspe$ during a step,
			the second element is fixed and only the first element (in $\B$) may 
			change.
			We split $X$ into $X^{=} = \{x \in X \mid x_{\ispe}^b = x_{\jspe}^b\}$
			and $X^{\neq} = \{x \in X \mid x_{\ispe}^b \neq x_{\jspe}^b\}$,
			and observe the following facts by the definition of $f_{\ispe}$ and 
			$f_{\jspe}$:
			\begin{itemize}
			\item Under $\mu$ and $\mu'$, all the elements of $X^=$ are fixed points
				(indeed, only $x_{\ispe}^b$ and $x_{\jspe}^b$ may evolve by copying the 
				other).
			\item Under $\mu$, let $m, m'$ be the respective number of times $\ispe,
				\jspe$ have been updated prior to the $k$-th block of $\varphi(\mu)$ in 
				which they are updated synchronously.
				Consider the configurations $x, y \in X^{\neq}$ with $x_{\ispe} = 
				(0, -m \mod \frac{p}{|S|})$, $x_{\jspe} = (1, -m' \mod \frac{p}{|S'|})$,
				$y_{\ispe} = (1, -m \mod \frac{p}{|S|})$ and $y_{\jspe} = (0, -m' \mod	
				\frac{p}{|S'|})$.
				It holds that $\mmjoblock{f}{\mu}(x) = y$ and 
				$\mmjoblock{f}{\mu}(y) = x$, because $x_{\ispe}^b$ and $x_{\jspe}^b$ are 
				exchanged synchronously when $x_{\ispe}^\ell = x_{\jspe}^\ell = 0$ 
				during the $k$-th block of $\varphi(\mu)$, and are not exchanged again 
				during that step by the choice of the modulo.
				Hence, $\mmjoblocklimit{f}{\mu}$ has a limit cycle of length two.
			\item Under $\mu'$, for any $x \in X^{\neq}$, there is a substep with 
				$x_{\ispe}^\ell = 0$ and there is a substep with $x_{\jspe}^\ell = 0$,
				but they are not the same substep (because $\ispe$ and $\jspe$ are never 
				synchronised in $\mu'$).
				As a consequence, $x_{\ispe}^b$ and $x_{\jspe}^b$ will end up having the 
				same value (the first to be updated copies the bit from the second, then 
				the second copies its own bit), \emph{i.e.}~$\mmjoblock{f}{\mu'}(x) \in 
				X^{=}$, and therefore $\mmjoblocklimit{f}{\mu'}$ has only fixed points.
		  \end{itemize}
			We conclude in this case that $\mmjoblocklimit{f}{\mu} \not \sim
			\mmjoblocklimit{f}{\mu'}$, because one has a limit cycle of length two, 
			whereas the other has only fixed points.

		\item[(2.2)] If the permutation preserves the columns within the matrices
	  (meaning that the automata within the same column in $\mu$ are also in the 
	  same column in $\mu'$), then we consider two last subcases:
	  	\begin{enumerate}
	  	\item[(2.2.1)] Moreover, if the permutation of some matrix is not 
	  		circular (meaning that there are three columns which are not in the same 
	  		relative order in $\mu$ and $\mu'$), then there are three automata 
	  		$\ispe$, $\jspe$ and $\kspe$ in the same matrix such that in $\mu$, 
	  		automaton $\ispe$ is updated first, then $\jspe$, then $\kspe$; whereas
				in $\mu'$, automaton $\ispe$ is updated first, then $\kspe$, then 
				$\jspe$. 
				Let us consider the automata network $f$ such that:
				\begin{itemize}
				\item $X = \B^n$;
				\item $f_{\ispe}(x) = x_{\kspe}$, $f_{\jspe}(x) = x_{\ispe}$ and 
					$f_{\kspe}(x) = x_{\jspe}$; and
				\item $f_i(x) = x_i$ if $i \notin \{\ispe, \jspe, \kspe\}$.
				\end{itemize}
				If the three automata are updated in the order $\ispe$ then $\jspe$ 
				then $\kspe$, as it is the case with $\mu$, then after any update, 
				they will all have taken the same value.
				It implies that $\mmjoblock{f}{\mu}$ has only fixed points,
				precisely the set 
				$P = \{x \in \B^n \mid x_{\ispe} = x_{\jspe} = x_{\kspe}\}$.

				If they are updated in the order $\ispe$ then $\kspe$ then $\jspe$, as
				with $\mu'$, however, the situation is a bit more complex.
				We consider two cases, according to the number of times they are updated 
				during a period (recall that since they belong to the same matrix,
				they are updated repeatedly in the same order during the substeps):
				\begin{itemize}
				\item If they are updated an odd number of times each, then automata $
					\ispe$ and $\jspe$ will take the initial value of automaton $\kspe$, 
					and automaton $\kspe$ will take the initial value of automaton 
					$\jspe$.
					In this case, $\mmjoblocklimit{f}{\mu'}$ has the fixed points $P$
					and limit cycles of length two.
				\item If they are updated an even number of times each, then the reverse 
					will occur: 
					automata $\ispe$ and $\jspe$ will take the initial value of automaton 
					$\jspe$, and automaton $\kspe$ will keep its initial value.
					In this case, $\mmjoblocklimit{f}{\mu'}$ has the fixed points 
					$Q = \{x \in \B^n \mid x_{\ispe} = x_{\jspe}\}$ which strictly 
					contains $P$ (\emph{i.e.}~$P \subseteq Q$ and $Q \setminus P \neq
					\emptyset$).
				\end{itemize}
				In both cases $\mmjoblocklimit{f}{\mu'}$ has more than the fixed points 
				$P$ in its limit set, hence we conclude that
				$\mmjoblocklimit{f}{\mu} \not \sim \mmjoblocklimit{f}{\mu'}$.
	  	\item[(2.2.2)] Moreover, if the permutation of all matrices is circular,
				then we first observe that when $\varphi(\mu)$ and $\varphi(\mu')$ have 
				one block in common, they have all blocks in common (because of the 
				circular nature of permutations), \emph{i.e.}~$\mu \equiv_\star \mu'$.
				Thus, under our hypothesis, we deduce that $\varphi(\mu)$ and 
				$\varphi(\mu')$ have no block in common.
				As a consequence, there exist automata $\ispe, \jspe$ with the property 
				from case~(2.1), namely synchronised in a block of 
				$\varphi(\mu)$ but never synchronised in any block of $\varphi(\mu')$,
			  and the same construction terminates this proof.
			\end{enumerate}
		\end{enumerate}
  \end{enumerate}
\end{proof}

Let $\BPn^\star = \BPn / \equiv_\star$ denote the corresponding quotient set.

\begin{theorem}
	Let $\lcm(i) = \lcm(\{j \in \llbracket 1, d(i) \rrbracket \mid 
	m(i,j) \geq 1\})$.
	For any $n \geq 1$, we have:
	\[
		|\BPn^\star| = 
			\sum_{i = 1}^{p(n)} 
				\frac{n!}{\prod_{j = 1}^{d(i)} 
					\left( 
						m(i,j)! 
					\right)^j}
				\cdot 
				\frac{1}{\lcm(i)}\text{.}
  \]
\end{theorem}

\begin{proof}
	Let $\mu, \mu' \in \BPn$ two update modes such that $\mu \equiv \mu'$. 
	Then their sequential forms are of the same length, and each automaton 
	appears the same number of times in both of them. 
	This means that, if an automaton is in an o-block of size $k$ in $\mu$'s 
	partitioned order form, then it is also in an o-block of the same size in 
	$\mu'$'s. 
	We deduce that two update modes of size $n$ can only be equivalent as defined 
	in Definition~\ref{def:dyn_iso} if they are generated from the same partition 
	of $n$.

	Let $\mu \in \BPn^0$, generated from partition $i$ of $n$.
	Then $\varphi(\mu)$ is of length $\lcm(i)$. 
	Since no two elements of $\BPn^0$ have the same block-sequential form, the 
	equivalence class of $\mu$ in $\BPn^0$ contains exactly $\lcm(i)$ elements, 
	all generated from the same partition $i$ (all the blocks of $\varphi(\mu)$ 
	are different).
	Thus, the number of elements of $\BPn^\star$ generated from a partition $i$
	is the number of elements of $\BPn^0$ generated from partition $i$,
	divided by the number of elements in its equivalence class for $\BPn^\star$, 
	namely $\lcm(i)$.
\end{proof}

\begin{remark}\normalfont\label{remark:BPnstar}
  The formula for $|\BPn^\star|$ can actually be obtained from any formula in
  Theorem~\ref{theorem:BPn0} by multiplying by $\frac{1}{\lcm(i)}$ inside the 
  sum on partitions (from $i = 1$ to $p(n)$).
\end{remark}

While counting the elements of $\BPn^\star$ was pretty straightforward,
enumerating them by ensuring that no two partitioned orders are the same up 
to circular permutation of their block-sequential rewritings 
(Definition~\ref{def:dyn_iso}) is more challenging.
This is performed by Algorithm~\ref{algo:BPiso}.
It works much like Algorithm~\ref{algo:BP0}, except for the following differences. 
In the first function, \texttt{EnumBPiso}:
right after choosing the partition, a list of coefficients $a[j]$ is determined, 
in a modified algorithm that computes $\lcm(i)$ inductively (lines 3-10). 
These coefficients are used in the second auxiliary function 
\texttt{EnumBlockIsoAux}, where the minimum $min_j$ of the matrix $M_j$ is 
forced to be in the first $a[j]$ columns of said matrix (lines 35-37, the 
condition is fulfilled when $min_j$ has not been chosen within the $a[j] - 1$ 
first columns, then it is placed in that column so only $m(i,j) - 1$ elements 
are chosen).
The correction of Algorithm~\ref{algo:BPiso} is argued in the following 
statement.

\begin{algorithm}[t!]
  {\small \caption{Enumeration of $\BPn^\star$}
  \label{algo:BPiso}
  \DontPrintSemicolon
  \SetKw{KwAnd}{and}
  \SetKwProg{Fn}{Function}{:}{}
  \SetKwFunction{FEnumBP}{EnumBPiso}
  \SetKwFunction{FEnumBPaux}{EnumBPisoAux}
  \SetKwFunction{FEnumBlock}{EnumBlockIso}
  \SetKwFunction{FEnumBlockAux}{EnumBlockIsoAux}
  \SetKwComment{Comment}{\# }{}
	\bigskip  
  
  \Fn{\FEnumBP{$n$}}{
    \ForEach{$i \in \emph{partitions}(n)$}{
      $a$ is a list of size $d(i)$\;
      $b \leftarrow 1$\;
      \For{$j \leftarrow d(i)$ \KwTo $1$}{
        \If{$m(i,j) > 0$}{
          $a[j] \leftarrow \gcd(b, j)$\;
          $b \leftarrow \lcm(b, j)$\;
        }
        \Else{
          $a[j] \leftarrow j$\;
        }
      }
      \FEnumBPaux{$n$, $i$, $1$, $a$}
    }
  }\medskip
  
  \Fn{\FEnumBPaux{$n$, $i$, $j$, $a$, $M_1$, \dots, $M_{j-1}$}}{
    \If{$d(i) < j$}{
      \texttt{enumerate(}$M$\texttt{)}\;
      \Return
    }
    \If{$m(i, j) > 0$}{ 
      \ForEach{combination $A$ of size $j \cdot m(i, j)$ among
				$\entiers{n} \setminus \bigcup_{k = 1}^{j - 1} M_k$}{
        $min_j \leftarrow min(A)$\;
        \ForEach{$M_j$ enumerated by \FEnumBlock{$A$, $j$, $m(i,j)$,
					$min_j$, $a \emph{[} j \emph{]}$}}{
          \FEnumBPaux{$n$, $i$, $j+1$, $a$, $M_1$, \dots, $M_{j-1}$, $M_j$}
        }
      }
    }
    \Else{
      \FEnumBPaux{$n$, $i$, $j+1$, $a$, $M_1$, \dots, $M_{j-1}$, $\emptyset$}
    }
  }\medskip

  \Fn{\FEnumBlock{$A$, $j$, $m$, $min_j$, $a_j$}}{
    \ForEach{$C$ enumerated by \FEnumBlockAux{$A$, $m$, $min_j$, $a_j$}}{
      \For{$k \leftarrow 1$ \KwTo $m$}{
        \For{$\ell \leftarrow 1$ \KwTo $j$}{
          $M_j[k][\ell] = C_\ell[k]$\;
        }
      }
      \texttt{enumerate(}$M_j$\texttt{)}\;
      \Return
    }
  }\medskip

  \Fn{\FEnumBlockAux{$A$, $j$, $m$, $min_j$, $a_j$, $C_1$, ..., $C_\ell$}}{
    \If{$A = \emptyset$}{
      \texttt{enumerate(}$C$\texttt{)}\;
      \Return
    }
    \Else{
      \If{$|A| = m \cdot(j - a_j + 1)$ \KwAnd $min_j \in A$}{
        \ForEach{combination $B$ of size $m - 1$ among $A \setminus \{min_j\}$}{
          \FEnumBlockAux{$A \setminus (B \cup \{min_j\})$, $m$, $min_j$, $a_j$,
						$C_1$, \dots, $C_\ell$, $(B \cup \{min_j\})$}\;
        }
      }
      \Else{
        \ForEach{combination $B$ of size $m$ among $A$}{
          \FEnumBlockAux{$A \setminus B$, $m$, $min_j$, $a_j$, $C_1$, \dots,
						$C_\ell$, $B$}\;
        }
      }
    }
  }}
\end{algorithm}

\begin{theorem}
 Algorithm~\ref{algo:BPiso} is correct: it enumerates $\BPn^\star$.
\end{theorem}

\begin{proof}
  We first argue that Algorithm~\ref{algo:BPiso} enumerates the correct number
  of block-parallel update modes, and then that any pair $\mu, \mu'$
  enumerated is such that $\mu \not \equiv_\star \mu'$.

	For ease of reading in the rest of this proof, we will denote $a[j]$ as $a_j$
	and $m(i, j)$ as $m_{ij}$.
	From the placement of $min_j$ described above (forced to be within the first 
	$a_j$ columns of matrix $M_j$), the difference with Algorithm~\ref{algo:BP0} 
	is that, instead of having $\prod_{\ell = 1}^{j} \binom{(j - \ell + 1) \cdot
	m_{ij}}{m_{ij}}$ ways of filling matrix $M_j$, we only have the following
	number of ways (recall that $M_j$ has $j$ columns and $m_{ij}$ rows):
	\[
		\sum_{k = 1}^{a_j} 
			\left( 
				\prod_{\ell = 1}^{k - 1} 
					\binom{(j - \ell + 1) \cdot m_{ij} - 1}{m_{ij}} 
			\right) 
			\cdot 
			\binom{(j - k + 1) \cdot m_{ij} - 1}{m_{ij} - 1} 
			\cdot 
			\left(
				\prod_{\ell = k + 1}^{j}
					\binom{(j - \ell + 1) \cdot m_{ij}}{m_{ij}}
			\right)\text{.}
	\]
	Indeed, the formula above sums, for each choice of a column $k$ from $1$ to 
	$a_j$ where $min_j$ will be placed, the number of ways to place some elements 
	within columns $1$ to $k-1$ (first product on $\ell$), times the number of 
	ways to choose some elements that will accompany $min_j$ within column $k$ 
	(middle binomial coefficient), times the number of ways to place some other 
	elements within the remaining columns $k+1$ to $j$ (second product on $\ell$).
	Now, we have 
	$\binom{(j - k + 1) \cdot m_{ij}-1}{m_{ij} - 1} 
		= 
			\frac{m_{ij}}{(j - k + 1) \cdot m_{ij}} 
			\cdot 
			\binom{(j - k + 1) \cdot m_{ij}}{m_{ij}} 
		= 
			\frac{1}{(j - k + 1)} \cdot \binom{(j - k + 1) 
			\cdot 
			m_{ij}}{m_{ij}}$. 
	We also have
	$\binom{(j - \ell + 1) \cdot m_{ij} - 1}{m_{ij}}
		=
			\frac{j - l}{j - (l - 1)} 
			\cdot
			\binom{(j - \ell + 1) \cdot m_{ij}}{m_{ij}}$.
	This means that the sum of the possible ways to choose the content
	of matrix $M_j$ can be rewritten as follows:
	\begin{multline*}
		\sum_{k = 1}^{a_j}
			\left(
				\frac{1}{(j - k + 1)} \cdot \prod_{\ell = 1}^{k - 1}
					\frac{j - \ell}{j - (\ell - 1)} \cdot \prod_{\ell = 1}^j 
						\binom{(j - \ell + 1) \cdot m_{ij}}{m_{ij}}
			\right)\\ 
		\begin{split}
			\;=\; & 
				\sum_{k = 1}^{a_j}
					\left(
						\frac{1}{(j - k + 1)} \cdot \frac{j - k + 1}{j} \cdot 
						\prod_{\ell = 1}^j
							\binom{(j - \ell + 1) \cdot m_{ij}}{m_{ij}}
					\right)\\
			\;=\; & 
				\sum_{k = 1}^{a_j} 
					\left(
						\frac{1}{j} \cdot \prod_{\ell = 1}^j 
							\binom{(j - \ell + 1) \cdot m_{ij}}{m_{ij}}
					\right)\\
		  \;=\; & 
		  	\frac{a_j}{j} \cdot \prod_{\ell = 1}^j
		  		\binom{(j - \ell + 1) \cdot m_{ij}}{m_{ij}}\text{.}
		\end{split}
	\end{multline*}
	Hence in total, the algorithm enumerates the following number of update modes:
	\[
		\sum_{i = 1}^{p(n)}
			\prod_{j = 1}^{d(i)}
				\left(
					\binom{n - \sum_{k = 1}^{j - 1} k \cdot m_{ik}}{j \cdot m_{ij}}
					\cdot
					\frac{a_j}{j} \cdot \prod_{\ell = 1}^{j}
						\binom{(j - \ell + 1) \cdot m_{ij}}{m_{ij}}
    		\right)\text{.}
  \]
	In order to prove that this number is equal to $|\BPn^\star|$, we need to 
	prove that $\prod_{j = 1}^{d(i)} \frac{a_j}{j} = \frac{1}{\lcm(i)}$.
	Denoting $L(j) = \lcm(\{k \in \entiers{j, d(i)} \mid m_{ik} > 0\})$, we prove 
	by induction that at the end of each step of the \textbf{for} loop from lines 
	5-10, we have:
	\[
		\prod_{k = j}^{d(i)} 
			\frac{a_k}{k} 
		= 
			\frac{1}{L(j)}\text{,}
	\] 
	and the claim follows (when $j = 1$, we get $L(j) = \lcm(i)$).
	At the first step, $j = d(i)$, and 
	\[
		\frac{a_{d(i)}}{d(i)} 
		= 
			\frac{\gcd(\{d(i), 1\})}{d(i)} 
		= 
			\frac{1}{L(j)}\text{.}
	\]
	We assume as induction hypothesis that for a given $j$, we have
	$\prod_{k = j}^{d(i)}\frac{a_k}{k} = \frac{1}{L(j)}$.
	There are two possible cases for $j - 1$:
	\begin{itemize}
	\item If $m_{i(j-1)} = 0$, then 
		\[
			a_{j-1} = j-1 \text{ and }
			\prod_{k = j - 1}^{d(i)}
				\frac{a_k}{k} 
			= 
				\frac{j - 1}{(j - 1)\cdot L(j)} 
			= 
				\frac{1}{L(j-1)}\text{.}
		\]
	\item Otherwise, 
		\[
			\prod_{k = j - 1}^{d(i)}
				\frac{a_k}{k} 
			= 
				\frac{\gcd(\{j - 1, L(j)\})}{(j - 1) \cdot L(j)}\text{.}
		\]
		And since $\frac{a \cdot b}{\gcd(\{a, b\})} = \lcm(\{a, b\})$,
		we have
		\[
			\prod_{k = j - 1}^{d(i)}
				\frac{a_k}{k} 
			= 
				\frac{1}{\lcm(\{j-1, L(j)\})} 
			=
		    \frac{1}{L(j-1)}\text{.}
		\]
	\end{itemize}
	We conclude that at the end of the loop, we have $\prod_{j = 1}^{d(i)}
		\frac{a_j}{j} = \frac{1}{\lcm(i)}$, and thus that the algorithm enumerates 
                the following number of update modes (cf.~Remark~\ref{remark:BPnstar}):
	\[
		\sum_{i = 1}^{p(n)} 
			\prod_{j = 1}^{d(i)}
				\left(
					\binom{n - \sum_{k = 1}^{j-1} k \cdot m_{ik}}{j \cdot m_{ij}}
					\cdot
					\prod_{\ell = 1}^{j}
						\binom{(j - \ell + 1) \cdot m_{ij}}{m_{ij}}
				\right) 
			\cdot 
			\frac{1}{\lcm(i)} 
		= 
			|\BPn^\star|\text{.}
	\]\medskip

We now need to prove that the algorithm does not enumerate two equivalent
update modes (in the sense of $\equiv_\star$).
Algorithm~\ref{algo:BPiso} is heavily based on Algorithm~\ref{algo:BP0}, in
such a way that, for a given input, the output of Algorithm~\ref{algo:BPiso} 
will be a subset of that of Algorithm~\ref{algo:BP0}.
Algorithm~\ref{algo:BP0} enumerates $\BPn^0$, which implies that every update
mode enumerated by it has a different image by $\varphi$. 
This means that every block-parallel update mode enumerated by 
Algorithm~\ref{algo:BPiso} also has a different image by $\varphi$, and that the 
algorithm does not enumerate two equivalent update modes with a shift of $0$.

We now prove by contradiction that the algorithm does not enumerate two 
equivalent update modes with a non-zero shift. 
Let $\mu, \mu' \in \BPn$ be two update modes, both enumerated by 
Algorithm~\ref{algo:BPiso}, such that $\mu \equiv_\star \mu'$ with a non-zero 
shift. 
Then, there is $k \in \entiers[1]{|\varphi(\mu')| - 1}$ such that 
$\varphi(\mu) = \sigma^k(\varphi(\mu'))$, and $\mu, \mu'$ are both generated
from the same partition $i$, with $|\varphi(\mu)| = |\varphi(\mu')| = \lcm(i)$.
Moreover, for each $j \in \entiers{d(i)}$ the matrix $M_j$ must contain
the same elements $A$ in both $\mu$ and $\mu'$ (so that they are repeated every 
$j$ blocks in both $\varphi(\mu)$ and $\varphi(\mu')$), hence in particular 
$min_j$ is the same in both enumerations.

We will prove by induction that every $j \in \entiers{d(i)}$ such that 
$m(i, j) > 0$ divides $k$, and therefore $k = 0$ (equivalently $k = \lcm(i)$), 
leading to a contradiction.
For the base case $j = d(i)$, we have $a_{d(i)} = 1$ (first iteration of the 
\textbf{for} loop lines 5-10), hence the call of \texttt{EnumBlockIsoAux} for 
any set $A$ passes the condition of line 35 and $min_{d(i)}$ is immediately 
chosen to belong to $C_1$ (the first column of $M_{d(i)}$).
When converted to block-sequential update modes, it means that $min_{d(i)}$
appears in all blocks indexed by $d(i) \cdot t$ with $t \in \N$, hence $k$ must 
be a multiple of $d(i)$ so that $\sigma^k$ maps blocks containing $min_{d(i)}$
to blocks containing $min_{d(i)}$.

As induction hypothesis, assume that for a given $j$, every $\ell \in 
\llbracket j, d(i)\rrbracket$ such that $m_{i\ell} > 0$ divides $k$. 
We will prove that $j'$, the biggest number in the partition $i$ 
(\emph{i.e.}~with $m_{ij'} > 0$) that is smaller than $j$, also divides $k$.
In the matrix $M_{j'}$, the minimum $min_{j'}$ is forced to appear within the 
$a_{j'}$ first columns.
This means that block indexes where it appears in both $\varphi(\mu)$ and 
$\varphi(\mu')$ can be written respectively as $j'\cdot t + b$ and 
$j'\cdot t + b'$ respectively, with $t \in \N$ and $b , b'\in 
\llbracket 1, a_{j'}\rrbracket$.
As a consequence, an automaton from $M_{j'}$ that is at the position $b$ in
$\varphi(\mu)$ is at a position of the form $j' \cdot t + b'$ in 
$\varphi(\mu')$. 
It follows that $b + k = j'\cdot t + b'$, which can be rewritten as 
$k = t \cdot j' + b' - b = t \cdot j' + d$, with $t \in \N$ and 
$d = b'- b \in \llbracket - a_{j'} + 1, a_{j'} - 1 \rrbracket$.
Moreover, we know by induction hypothesis that every number in the partition $i$
that is greater than $j'$ divides $k$, making $k$ a common multiple of these
numbers. 
We deduce that their lowest common multiple also divides $k$. 
Given that $a_{j'}$ is the $\gcd$ of $j'$ and said $\lcm$ (lines 7-8), it means 
that $a_{j'}$ divides both $j'$ and $k$, which implies that it also divides $d$.
Since $d$ is in $\llbracket -a_{j'} + 1, a_{j'} - 1 \rrbracket$, we have $d = 0$
and thus, $j'$ divides $k$. This concludes the induction.

If every number of the partition $i$ divides $k$ and $k \in \entiers{\lcm(i)}$,
then $k=0$, leading to a contradiction. 
This concludes the proof of correctness.
\end{proof}

\subsection{Implementations}
\label{s:implementations}

Proof-of-concept Python implementations of the three
Algorithms~\ref{algo:BP},~\ref{algo:BP0}, and~\ref{algo:BPiso}
are available on the following repository:
\begin{center}
  \url{https://framagit.org/leah.tapin/blockpargen}.
\end{center}
It is archived by Software Heritage at the following permalink:
\begin{center}
  \href{https://archive.softwareheritage.org/browse/directory/f1b4d83c854a4d042db5018de86b7f41ef312a07/?origin_url=https://framagit.org/leah.tapin/blockpargen}{https://archive.softwareheritage.org/browse/directory/\\f1b4d83c854a4d042db5018de86b7f41ef312a07/?origin\_url=\\https://framagit.org/leah.tapin/blockpargen}.
\end{center}
We have conducted numerical experiments on a standard laptop,
presented on Figure~\ref{fig:xp}.

\begin{figure}
\begin{minipage}{0.45\textwidth}
 \centering
 \begin{tabular}{| c | c | c | c |}
  \hline
  $n$ & $\BPn$ & $\BPn^0$ & $\BPn^\star$ \\
  \hline
  1 & 1 & 1 & 1\\
  & - & - & -\\
  \hline
  2 & 3 & 3 & 2\\
  & - & - & -\\
  \hline
  3 & 13 & 13 & 6\\
  & - & - & -\\
  \hline
  4 & 73 & 67 & 24\\
  & - & - & -\\
  \hline
  5 & 501 & 471 & 120\\
  & - & - & -\\
  \hline
  6 & 4051 & 3591 & 795\\
  & - & - & -\\
  \hline
  7 & 37633 & 33573 & 5565\\
  & - & 0.103s & -\\
  \hline
  8 & 394353 & 329043 & 46060\\
  & 0.523s & 0.996s & 0.161s\\
  \hline
  9 & 4596553 & 3919387 & 454860\\
  & 6.17s & 12.2s & 1.51s\\
  \hline
  10 & 58941091 & 47827093 & 4727835\\
  & 1min24s & 2min40s & 16.3s\\
  \hline
  11 & 824073141 & 663429603 & 54223785\\
  & 21min12s & 38min31s & 3min13s\\
  \hline
  12 & 12470162233 & 9764977399 & 734932121\\
  & 5h27min38s & 9h49min26s & 45min09s\\
  \hline
  \end{tabular}
\end{minipage}\hfill
\begin{minipage}{0.45\textwidth}
\centering
\includegraphics[scale=0.46]{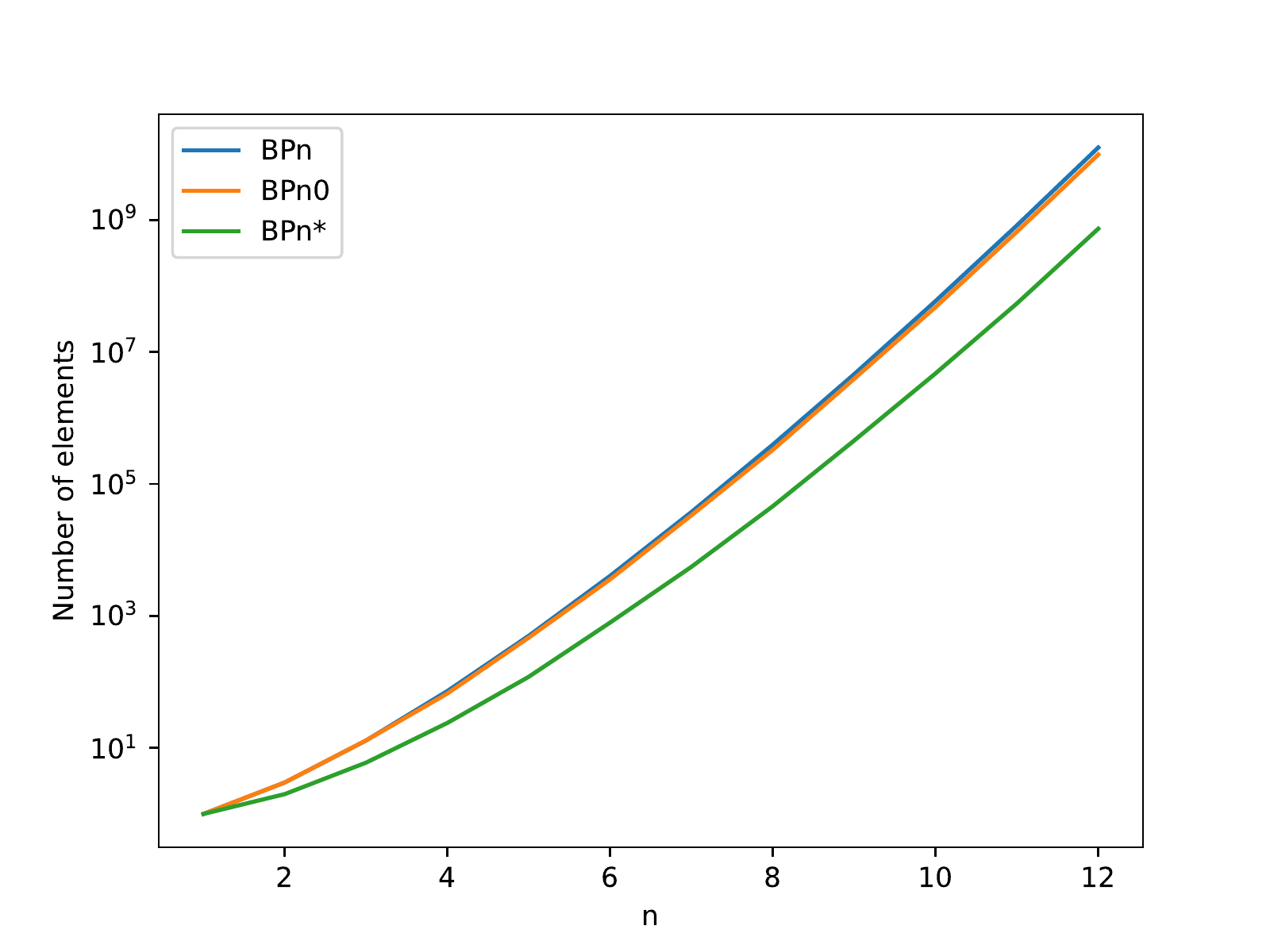}
\end{minipage}
\caption{
  Numerical experiments of our Python implementation of
  Algorithms~\ref{algo:BP},~\ref{algo:BP0}, and~\ref{algo:BPiso}
  on a standard laptop (processor Intel-Core$^\text{TM}$ i7 @ 2.80 GHz).
  For $n$ from $1$ to $12$, the table (left) presents the size of
  $\BPn$, $\BPn^0$ and $\BPn^*$ and running time to enumerate
  their elements (one representative of each equivalence class;
  a dash represents a time smaller than $0.1$ second),
  and the graphics (right) depicts their respective sizes
  on a logarithmic scale.
  Observe that the sizes of $\BPn$ and $\BPn^0$ are comparable,
  whereas an order of magnitude is gained with $\BPn^*$,
  which may be significant for advanced numerical experiments
  regarding limit dynamics under block-parallel udpate modes.
}
\label{fig:xp}
\end{figure}

\section{Conclusion and perspectives}
\label{s:conclusion}

This article presents a first theoretical study on the block-parallel update 
modes in the AN setting, and focuses on some of their combinatorial features.
In particular, beyond their general counting and enumeration from the set theory 
standpoint and a result clarifying their intersection with more classical 
block-sequential update modes, we give neat results of block-parallel update 
modes in connection with the AN-based dynamical systems that they can produce.
Indeed, by basing ourselves on combinatorial proofs about their number, we 
produce algorithmic enumerations of such modes in two cases: 
the first one aims at emphasising formally representatives of block-parallel 
update modes so that they allow to generate all the possible distinct underlying 
dynamical systems, 
and the second one follows the same vein and deals with the block-parallel 
update modes for which we have the guaranty that they do not generate dynamical 
systems having the same set of limit cycles (limit dynamics).

In this context of enumeration algorithms, one of the first natural perspectives 
would concern their complexity. 
First, whilst a proof is still required, 
Algorithms~\ref{algo:BP},~\ref{algo:BP0}, and~\ref{algo:BPiso} seem to belong to 
\textsf{EnumP}~\cite{J-Creignou2019}. 
Indeed, the outputs of these algorithms are of polynomial size regarding the 
number of automata (by definition of partitioned orders), and for each of these 
algorithms, deciding if a set of sublists belongs to their outputs are 
\emph{a priori} in \textsf{P}. 
However, a question remains: to which \textsf{EnumP} subclass$\cdot$es do they 
belong to? 
Intuitively, it seems that they should be in \textsf{DelayP} but this question 
still needs to be addressed.

The peculiarity of block-parallel update modes is that they give rise to local 
update repetitions during a period. 
This is indeed the case for all block-parallel update modes which are not 
block-sequential (\emph{i.e.}~modes with at least two blocks of distinct sizes 
when defined as a partitioned order, cf.~Lemma~\ref{lemma:BPequivBS}). 
Since we know that local update repetitions can break the fixed point invariance 
property which holds in block-sequential ANs (cf.~the example given in 
Remark~\ref{remark:repet}), it should be pertinent to characterise the 
conditions relating these repetitions to the architecture of interactions 
between the automata defined by the local functions under which the set of 
obtained fixed points is not that obtained with the parallel update mode. 
More generally, as a complement of the results of Section~\ref{s:BPstar}, it 
would be interesting to study the following problem: 
given an AN $f$, to which extent is $f$ block-parallel sensible/robust?
In~\cite{J-Aracena2009,J-Aracena2012,J-Aracena2013}, the authors addressed this 
question on block-sequential Boolean ANs by developing the concept of 
update digraphs which allows to capture conditions of dynamical equivalence at 
the syntactical level. 
However, this concept is not helpful anymore as soon as local update repetitions 
are at stake. 
Hence, creating a new concept of update digraphs in the general context of 
periodic update modes would be an essential step forward to explain and 
understand updating sensitivity/robustness of ANs.

Another track of research would be to understand how basic interaction cycles of 
automata evolve when they are updated in block-parallel. 
For instance, the authors of~\cite{C-Goles2010} have shown that such cycles in 
the Boolean setting are somehow very robust to block-sequential update modes 
variations: the number of their limit cycles of length $p$ is the same as that 
of a smaller cycle (of same sign) evolving in parallel. Together with the 
combinatorial analysis of~\cite{dns12}, this result provides a complete 
analysis of the asymptotic dynamics of Boolean interaction cycles. 
This gives rise to the following question: do interaction cycles behave 
similarly depending on block-parallel update modes variations?
Here also, the local update repetitions should play an essential role. 
Such a study could constitute a first approach of the more general problem 
evoked above, since it is well known that cycles are the behavioural complexity 
engines of ANs~\cite{B-Robert1986}. 

Eventually, since block-parallel update modes form a new family of update modes 
and since the field of investigation related to them is therefore still 
completely open today, we think that a promising perspective of our work would 
consist in dealing with the complexity of classical decision problems for ANs, 
in the lines of~\cite{J-Dennunzio2019} about reaction systems. 
The general question to be addressed here is: 
do local update repetitions induced by block-parallel update modes make such 
decision problems take place at a higher level in the polynomial hierarchy,
or even reach polynomial space completeness?

\paragraph{Acknowledgements}

The authors were funded mainly by their salaries as French State agents. This 
work has furthermore been secondarily supported by ANR-18-CE40-0002 FANs project 
(KP \& SS), ECOS-Sud C19E02 SyDySy project (SS), and STIC AmSud 22-STIC-02 CAMA 
project (KP, LT \& SS).

\bibliographystyle{plain}
\bibliography{biblio}

\end{document}